\documentclass[preprint, 3p]{elsarticle}

\usepackage[english]{babel}
\usepackage[utf8]{inputenc}
\usepackage{amsthm}
\usepackage{amsmath}
\usepackage{amssymb}
\usepackage{hyperref}
\usepackage{enumerate}
\usepackage{xcolor}
\usepackage{color,soul}
\usepackage{tikz}
\usepackage{placeins}
\usepackage{subfigure} 
\usepackage{booktabs}
\usepackage{bbm}

\journal{TBA}

\newtheorem{theorem}{Theorem}[section]
\newtheorem{corollary}[theorem]{Corollary}
\newtheorem{lemma}[theorem]{Lemma}
\newtheorem{proposition}[theorem]{Proposition}
\newtheorem{definition}[theorem]{Definition}
\newtheorem{remark}[theorem]{Remark}
\numberwithin{equation}{section}

\DeclareMathOperator{\esssup}{ess\,sup}

\DeclareMathOperator{\spt}{spt}
\DeclareMathOperator{\argmax}{argmax}
\DeclareMathOperator{\diag}{diag}

\begin{document}
	
\begin{frontmatter}
	
\title{On finite population games of optimal trading}

\author[EMApFGV]{D. Evangelista}
\ead{david.evangelista@fgv.br}

\author[IMEUFF]{Y. Thamsten}
\ead{ythamsten@id.uff.br}

\address[EMApFGV]{Escola de Matem\'atica Aplicada (EMAp), Funda\c{c}\~ao Get\'ulio Vargas (FGV), 22250-900, Rio de Janeiro, RJ, Brasil}

\address[IMEUFF]{Instituto de Matem\'atica e Estat\'istica (IME), Universidade Federal Fluminense (UFF), 24020-140, Niter\'{o}i, RJ, Brasil}

\date{}
	
	
\begin{abstract}
    We investigate stochastic differential games of optimal trading comprising a finite population. There are market frictions in the present framework, which take the form of stochastic permanent and temporary price impacts. Moreover, information is asymmetric among the traders, with mild assumptions. For constant market parameters, we provide specialized results. Each player selects her parameters based not only on her informational level but also on her particular preferences. The first part of the work is where we examine the unconstrained problem, in which traders do not necessarily have to reach the end of the horizon with vanishing inventory. In the sequel, we proceed to analyze the constrained situation as an asymptotic limit of the previous one. We prove the existence and uniqueness of a Nash equilibrium in both frameworks, alongside a characterization, under suitable assumptions. We conclude the paper by presenting an extension of the basic model to a hierarchical market, for which we establish the existence, uniqueness, and characterization of a Stackelberg-Nash equilibrium.
\end{abstract}	

\begin{keyword}
Finite Population Games, Optimal Execution, Price Impacts, Hierarchic Games, Asymmetric Information. 
\MSC[2010] 91A06, 91A15, 91A80, 93E20.
\end{keyword}

\end{frontmatter}	

\section{Introduction}

It is often the case that large institutional investors have to execute large trades. For instance, when there is a market shock, it is common for these agents to diminish their exposure in certain assets to comply with regulatory requirements, see \cite{braouezec2018risk,wagalath2018liquidation}. In these circumstances, a plethora of issues arise. Here, we are mainly concerned with three of these, namely: (i) having to deal with market frictions; (ii) managing the risk stemming from the uncertainty in price movements; (iii) facing the presence of arbitrageurs trying to profit out of the pressure that sizeable trades exert on the price of the assets. 

Transaction costs can be intelligible in simple terms, such as brokerage firms' fees, or can have a more complex nature, such as indirect costs. A  popular research directions on the optimal trade execution problem in the context of markets with frictions began with the seminal works \cite{almgren2001optimal,bertsimas1998optimal} of Bertsimas and Lo, and Almgren and Chriss. It is also worthwhile to mention an alternative approach that Obizhaeva and Wang proposed, see \cite{obizhaeva2013optimal}. They introduce supply-demand functions for the Limit Order Book (LOB), deriving a price impact process for a LOB aspect called resilience. We can see the former approach as a particular case of the second one, namely, its high resilience limit. We will focus henceforth on the Almgren-Chriss (AC) setting. The literature in this direction is quite rich, e.g., see \cite{cartea2015algorithmic,gueant2016financial} and the references therein.  

The AC model is a phenomenological one, with costs stemming from limited liquidity, which manifests in two kinds of price impact: permanent and temporary. On the one hand, the effect of a given agent's trading rate in the dynamics of the asset price is what we understand as a permanent price impact. On the other hand, the temporary price impact refers to the additional cost per share that the investor incurs by consuming layers of liquidity within the LOB (a process often referred to as ``walking the book'') to execute an order fully. Empirically, there are assessments of the reasonableness of this model in \cite{cartea2015algorithmic,sophie2018market}.

There are several possible extensions of the AC model in the multiplayer setting. One possible direction is, e.g., to proceed as in \cite{bank2018liquidity,fujii2020mean} or in the case of the dealer market of \cite{bouchard2018equilibrium}, in which the asset's price is, in a certain sense, determined to be an equilibrium resulting from a market clearing condition. This approach seems more suited to investigate the problem from a price formation perspective. In another direction, we adopt a more phenomenological modeling viewpoint, relatively close to \cite{cardaliaguet2017mean,drapeau2019fbsde,fu2018mean,huang2019mean,luo2018nash,schied2017state}, or the open market of \cite{bouchard2018equilibrium}. More precisely, we assume that the asset's mid-price reacts to the aggregate action of the traders. This reaction happens in a detrimental way to the overall population movement, i.e., an aggregate sell (buy) pressure pushes the price down (up); see \cite{bucci2020co,cont2016institutional} for some discussions in this connection. The work \cite{drapeau2019fbsde} considers an $n-$person game where the drift and volatility of the asset price are stochastic, but the impact parameters are deterministic and there is no liquidation constraint. The paper \cite{luo2018nash} regards a similar game, but with constant parameters and transient impact, whereas \cite{schied2017state} studies the finite population game with terminal state constraint and also assumes constant market parameters.


The perspective of \cite{bouchard2018equilibrium,fujii2020mean} is to study equilibria determined by a market clearing condition. As a first step, they assume the price to be given, implying each player's corresponding behavior. Subsequently, they solve decoupled individual optimization problems. In the sequel, they determine the equilibrium price via the balance condition (at least asymptotically), assuming the players' previously derived individual actions. This approach leads to a fixed point problem, which is equivalent to solving a suitable coupled forward-backward stochastic differential equations (FBSDE) system, see \cite[Lemma 5.1, Theorem 5.2]{bouchard2018equilibrium} and \cite[Theorem 3.1]{fujii2020mean}, together with the discussion following the latter therein. Consequently, in both of these works, individually optimal strategies are the best responses to this equilibrium price. The model in \cite{bank2018liquidity} assumes that there are two markets where trading takes place, namely, the dealer and open ones. In the latter, the transaction price of each player consists of a martingale (expected future dividend payments) plus some liquidity costs stemming from permanent and temporary price impacts. In this market, traders accommodate to a Nash equilibrium, for which closed-form formulas are at hand. In the former, from the open market's resulting behavior, the price is determined by the previously described equilibrium methodology. 

Another line of research related to differential games of optimal trading is the one that employs Mean-Field Game (MFG) models. MFGs constitute a branch of game theory developed to study the behavior of large populations of competing rational players. On many occasions, they are useful precisely because the finite population counterpart is not quite tractable. They were introduced in the mathematical community independently by J-M. Lasry and P-L. Lions \cite{ll1,ll2,ll3}, and by M. Huang, P. Caines and R. Malham\'e  \cite{Caines2,Caines1}. There are some relevant advances on finite population trading games prior to the development of MFGs, such as \cite{brunnermeier2005predatory,carlin2007episodic}. Lately, the efforts on MFGs instigated many studies on their finite population counterparts; we refer to the works \cite{cardaliaguet2017mean,casgrain2019trading,fu2018mean,fu2018leader,fujii2020mean,huang2019mean}. In \cite{fujii2020mean}, the authors have to work with the MFG limit, since the market clearing condition at the finite population level is incompatible with their adaptability assumption on the given price. In this way, they must investigate this condition asymptotically. In \cite{casgrain2019trading,fu2018mean,fujii2020mean,huang2019mean}, authors analyze the MFG model and show that the agent’s best response to the optimal aggregate rate is an approximate Nash equilibrium in the finite population game.

We also mention some advances related to settings of asymmetric information. The MFGs in \cite{fu2018mean,fu2018leader,fujii2020mean} have a common noise component, and allow the presence of private information. Differing beliefs between sub-populations of traders is a feature analyzed in \cite{casgrain2020mean}. Authors in \cite{bayraktar2018mini} propose a finite population game in which the drift is a latent process, and there is a temporary price impact, but their goal is to investigate mini-flash crashes. Other references considering latent factors are \cite{bismuth2019portfolio,casgrain2019trading,firoozi2018mean}. Equilibrium prices in the setting of finite population models with players having private information, which they bring into the game through their trading targets, is studied in \cite{choi2018smart}.

On the more technical side, we refer to the recent results in \cite{djete2020mean} regarding convergences of MFGs of controls and approximate Nash equilibria of the corresponding finite population counterparts, in both directions, considering solutions of the former in a proper sense. Proving the convergence of unconstrained games to constrained ones via monotonicity arguments is developed in \cite{ankirchner2014bsdes,graewe2015non,kim2013backward,kruse2016minimal,popier2006backward}. There is an alternative method which consists in identifying the precise asymptotic behavior of the candidate solution at the terminal time, see \cite{graewe2017optimal,graewe2018smooth}. In \cite{bensoussan2015maximum}, there is an analysis of stochastic Stackelberg differential games within symmetric information framework and Brownian filtrations.


Our basic model generalizes the finite population one described at the beginning of \cite[Subsection 1.2]{fu2018mean}, which in turn is motivated by \cite{cardaliaguet2017mean,carmona2015probabilistic}. In contradistinction to those, we do not assume that martingales driving the asset's price are arithmetic Brownian motions, and we also allow for the presence of an uncertain drift. On the one hand, as opposed to the previously mentioned finite population trading games, e.g., \cite{bank2018liquidity,bouchard2018equilibrium}, we allow all parameters to be stochastic, with mild assumptions. On the other hand, we also provide new results in the settings in which parameters are constant. In comparison with \cite{fu2018mean,fu2018leader,fujii2020mean}, our assumptions on informational asymmetry are more lenient. The extension to the leader-follower setting builds upon the ideas of \cite{fu2018leader,huang2019mean}.

We divide the present work into three parts. Firstly, we consider the unconstrained setting. That is a context in which players will not necessarily execute their total inventory by terminal time. However, they penalize strategies reaching terminal time with a non-zero amount of shares. We characterize the Nash equilibrium (NE) as the solution of a coupled FBSDE system of the McKean-Vlasov type. We prove that, under a weak interaction assumption, akin to that made in \cite{fu2018mean,horst2005stationary}, this FBSDE admits a unique solution. The condition we stipulate is equivalent to the one made in \cite{fu2018mean}, provided that the population size is sufficiently large. We use a continuation technique developed in \cite{peng1999fully}.

Under the assumption of constant parameters, but still heterogeneous, we prove that the NE rates, together with their corresponding inventories, form a solution of an ordinary differential equations (ODE) system. We demonstrate that, still under weak interaction, this ODE has a unique solution. Furthermore, we derive it in a semi-explicit form. If we further assume that parameters are homogeneous throughout the population, we show that the average inventory solves a second-order scalar ODE, akin to its MFG counterpart, derived in \cite{cardaliaguet2017mean}. For this ODE system, closed-form formulas are available. 

Secondly, we analyze the constrained problem, in which we require strict liquidation for all players. We prove that a similar characterization of the NE holds in this circumstance. Assuming weakly interacting agents, we manage to prove boundedness on the players' strategies uniformly on the terminal penalization parameters. Using weak convergence arguments, we show that we can pass to a subsequence to identify a solution to the FBSDE; hence, it is a NE for the constrained problem. When this solutions turns out to have continuous paths, a characterization follows. Putting ourselves under the same framework of \cite{fu2018mean}, which studies the MFG counterpart of our model, we prove that the average of the rates forming the NE converges to the optimal mean-field aggregation rate, as the population size tends to infinity, under suitable assumptions. We also provide a convergence rate. 

Thirdly, we develop an extension of our previous model to a hierarchical market. We assume there is a leader and a population of followers. We analyze a setting which generalizes that serving as motivation to the MFGs treated in \cite{fu2018leader,huang2019mean}. In our model, we assume that information is entirely asymmetric. Furthermore, we need not assume the leader’s strategy and parameters' adaptedness to the follower’s filtrations. In \cite{fu2018leader}, authors consider followers as informed traders, whereas \cite{huang2019mean} assumes no informational asymmetry. 

We introduce hierarchy by stipulating that the leader has a first-mover advantage. Therefore, the natural equilibrium to seek is that of Stackelberg-Nash. Thus, for each leader strategy, followers accommodate in an NE. Subsequently, the leader player solves an optimization problem conditional on minors following the corresponding NE. We prove that there exists a unique Stackelberg-Nash equilibrium, for given initial data, and characterize it by an FBSDE system consisting of the one identified in the previous part coupled with adequate adjoint states. If we assume that parameters are constant and homogeneous among the followers' population, we render this resulting FBSDE as a second-order three dimensional ODE system for the average state and adjoint variables.


We finish this Introduction by fixing some notations we use throughout the paper. In Section \ref{sec:the_market_model}, we describe our model, stipulate standing assumptions, and pose the equilibria problems we will investigate. We analyze the NE of the unconstrained problem in Section \ref{sec:NplayerGame}, proving existence and uniqueness. We also provide, in this Section, specialized results in the context of constant parameters. Next, in Section \ref{sec:AnalysisNplayerConstrained}, we obtain the NE of the constrained problem as an asymptotic weak limit of unconstrained NE, and relate the finite population game with its MFG counterpart. In Section \ref{sec:remarks_major_minor}, we extend our previous model to a hierarchic game of optimal trader, with a single major agent and a finite population of minor ones. We make concluding remarks in Section \ref{sec:conclusions}.

\textbf{Notations.} We consider a fixed time horizon $T>0,$ a population size $N\geqslant 1,$ and a complete filtered probability space $\left( \Omega, \mathcal{F}, \mathbb{F} = \left\{\mathcal{F}_t\right\}_{0\leqslant t \leqslant T}, \mathbb{P} \right),$ where $\mathbb{F}$ is complete, continuous, and such that $\mathcal{F}_T \subseteq \mathcal{F}.$ We write $\mathcal{N} := \left\{1,...,N\right\}$, and we assume that, to each $i \in \mathcal{N},$ there corresponds a filtration $\mathbb{F}^i = \left\{ \mathcal{F}^i_t \right\}_{0\leqslant t \leqslant T},$ which we also suppose complete and continuous, and satisfying $\mathcal{F}^i_t \subseteq \mathcal{F}_t,$ for every $t \in \left[0,T\right].$

From now on, $\mathcal{G}$ represents an arbitrary $\sigma-$algebra contained in $\mathcal{F}.$ We will consider the following functional spaces:
$$
L^2\left(\Omega,\mathcal{G}\right) := \left\{ X : X \text{ is } \mathcal{G}-\text{measurable and } \mathbb{E}\left[X^2\right] < \infty \right\};
$$
\begin{align*}
\mathbb{L}^2 := \Bigg\{ x=\left\{x_t\right\}_{0\leqslant t \leqslant T} : x \text{ is } \mathbb{F}-\text{progressively } &\text{measurable, and } \mathbb{E}\left[\int_0^T x_t^2\,dt\right] < \infty \Bigg\};
\end{align*}
$$
\mathbb{S} := \Bigg\{ x=\left\{x_t\right\}_{0\leqslant t \leqslant T} \in \mathbb{L}^2 :  \mathbb{E}\left[\sup_{0\leqslant t \leqslant T} x_t^2 \right] < \infty \Bigg\};
$$
\begin{align*}
    \mathcal{M}_i := \Big\{ M = \left\{ M_t \right\}_{0\leqslant t \leqslant T} : M_t &\in L^2\left(\Omega,\mathcal{F}^i_t\right),\, \text{ for a.e. } t \in \left[0,T\right], \text{ and } \left\{ M_t,\mathcal{F}^i_t\right\}_{0\leqslant t \leqslant T} \text{ is a martingale} \Big\};
\end{align*}
\begin{align*}
    \mathbb{M}_i := \Big\{ M \in \mathcal{M}_i : \mathbb{E}\left[ \left\langle M \right\rangle_T \right] < \infty \Big\}.
\end{align*}
Above, we have written $\left\langle M \right\rangle_T$ to denote the quadratic variation of $M$ over the interval $\left[0,T\right].$ We emphasize that we consider all the expectations appearing above under the measure $\mathbb{P}.$

These spaces are endowed with the norms
$$
\|X\|_{L^2\left(\Omega,\mathcal{G}\right)} := \mathbb{E}\left[X^2\right]^{1/2},
$$
$$
\|x\|_{\mathbb{L}^2} := \mathbb{E}\left[\int_0^T x_t^2\,dt\right]^{1/2},
$$
$$
\|y\|_{\mathbb{S}} := \mathbb{E}\left[\sup_{0\leqslant t \leqslant T} y_t^2 \right]^{1/2},
$$
and
$$
\|M\|_{\mathbb{M}_i} := \mathbb{E}\left[\left\langle M\right\rangle_T\right]^{1/2}.
$$
for each $X \in L^2\left(\Omega,\mathcal{G}\right),\, x \in \mathbb{L}^2,\, y \in \mathbb{S}$ and $M \in \mathbb{M}_i,\, i \in \mathcal{N}.$  For simplicity, we write from now on $\|\cdot\| := \|\cdot\|_{\mathbb{L}^2}.$ We clarify that the norms above are well-defined because we do not distinguish processes equal $dt\times d\mathbb{P}-$a.e.a.s. We abbreviated the expressions ``almost everywhere'' and ``almost surely''  by ``a.e.'' and ``a.s.,'' respectively. We will do this from now on. Similarly, we do not make a difference in random variables that coincide $\mathbb{P}-$a.s. If there is a version of a stochastic process with continuous paths, then this is the one we fix.

In general, throughout this work, given $m \geqslant 1$ normed spaces $E_1,...,E_m,$ we will consider in the product space $E:=\Pi_{i=1}^m E_i$ the norm
$$
\|\boldsymbol{x}\|_{E} := \left( \|x_1\|_{E_1}^2 + \cdots + \|x_m\|_{E_m}^2 \right)^{1/2},
$$
for $\boldsymbol{x} = \left(x_1,...,x_m\right)^\intercal \in E.$ For instance, we set
$$
\mathbb{M}_{(N)} := \Pi_{i=1}^N \mathbb{M}_i.
$$
With a slight abuse of notation, we will also denote the norm of $\left( \mathbb{L}^2 \right)^N$ by $\|\cdot\|.$ 

For $t \in \left[0,T\right],$ we denote by $\mathcal{P}^i_t$ the $L^2\left( \Omega, \mathcal{F}\right)-$projection operator onto $L^2\left(\Omega,\mathcal{F}^i_t\right),$ i.e.,
$$
\mathcal{P}^i_t\left(X\right) = \mathbb{E}\left[X|\mathcal{F}^i_t\right] \hspace{1.0cm} \left(X \in L^2\left(\Omega,\mathcal{F}\right) \right).
$$
We set $\boldsymbol{\mathcal{P}}_t : L^2\left(\Omega,\mathcal{F}\right)^N \rightarrow \Pi_{i=1}^N L^2\left(\Omega,\mathcal{F}^i_t\right)$ to be
$$
\boldsymbol{\mathcal{P}}_t\left(\boldsymbol{X}\right) := \left( \mathcal{P}^1_t\left(X^1\right),...,\mathcal{P}^N_t\left(X^N\right) \right)^\intercal,
$$
for $\boldsymbol{X} = \left( X^1,...,X^N \right)^\intercal.$ We also fix the following conventions, which we use throughout the remainder of the present work:
\begin{itemize}
    \item We convention that the letter $C$ will denote a generic positive constant, depending only on model parameters, which may change within estimates from line to line;
    \item For $d\geqslant 1$ and a matrix $\boldsymbol{\Lambda} \in \mathbb{R}^{d \times d},$ we will write $\boldsymbol{\Lambda} = \left( \ell_{ij} \right)_{i,j}$ to express that $\ell_{ij}$ is the entry $(i,j)$ of it. 
    \item We employ the notation $\diag\left( \ell_i \right) := \left( \delta_{ij}\ell_i \right)_{i,j}$ for diagonal matrices, where $\ell_1,...,\ell_d$ are given real numbers, and $\delta_{ij}$ denotes the Kronecker delta. We will particularly denote the square matrix of order $d,$ all of whose entries are equal to zero, by $\boldsymbol{0}^{d\times d}.$ We write $\boldsymbol{I}^{d \times d}$ to denote the identity matrix of order $d;$
    \item For two square matrices of the same order, $\boldsymbol{\Lambda}_1$ and $\boldsymbol{\Lambda}_2,$ we write $\boldsymbol{\Lambda}_1 \geqslant \boldsymbol{\Lambda}_2$ to signify that $\boldsymbol{\Lambda}_1 - \boldsymbol{\Lambda}_2$ is positive definite (not necessarily in the strict sense);
    \item For a given vector $\boldsymbol{a} = \left( a^1,...,a^N \right)^\intercal,$ we put
    $$
    \boldsymbol{a}^{-i} := \left( a^1,...,a^{i-1},a^{i+1},...,a^N\right)^\intercal,
    $$
    and
    $$
    \left( a^i, \boldsymbol{a}^{-i} \right) := \boldsymbol{a},
    $$
    for each $i \in \mathcal{N};$
    \item For two vectors $\boldsymbol{x},\boldsymbol{y} \in \mathbb{R}^d,$ we write $\boldsymbol{x} \cdot \boldsymbol{y} := \boldsymbol{x}^\intercal \boldsymbol{y};$
    \item We will only consider strong solutions of the stochastic differential equations (SDE), backward stochastic differential equations (BSDE), as well as of the FBSDE appearing in this text. We always understand solutions of an ODE system in the classical sense, except if we explicitly state otherwise.
\end{itemize}

\section{The market model}\label{sec:the_market_model}

Let us consider a stochastic differential game model comprising $N$ competitive rational traders negotiating a single financial asset. We index the players by $i \in \mathcal{N}.$ We will use the words trader, player or agent, interchangeably henceforth. Agent $i\in \mathcal{N}$ controls her trading rate $\left\{ \nu^i_t \right\}_{0\leqslant t \leqslant T} \in \mathbb{A}_i,$ where 
$$
\mathbb{A}_i := \left\{ \nu \in \mathbb{L}^2 : \nu \text{ is } \mathbb{F}^i-\text{progressively measurable} \right\}
$$
is the admissible set of trading strategies for this player. We write 
$$
\mathbb{A}_{(N)} := \Pi_{j=1}^N \mathbb{A}_j \text{ and } \mathbb{A}_{-i} := \Pi_{j\neq i} \mathbb{A}_j. 
$$
Similarly, we consider
$$
\mathbb{S}_i := \mathbb{S}\cap \mathbb{A}_i \text{ and } \mathbb{S}_{(N)} := \Pi_{i=1}^N \mathbb{S}_i.
$$
We endow each of the spaces $\mathbb{S}_i$ with the restriction of the norm of $\mathbb{S}$ to it.

Each player $i \in \mathcal{N}$ has a corresponding inventory process $\left\{q^i_t\right\}_{0\leqslant t \leqslant T} \in \mathbb{A}_i$ and a cash process $\left\{ x^i_t \right\}_{0\leqslant t \leqslant T} \in \mathbb{A}_i.$ We assume that the initial inventory $q^i_0,$ as well as the initial cash amount $x^i_0,$ belong to $L^2(\Omega,\mathcal{F}^i_0).$

\textbf{Dynamics of the state variables.} Let us fix $i \in \mathcal{N}$ arbitrarily. The inventory $\left\{q^i_t\right\}_{0\leqslant t \leqslant T}$ of the agent $i$ evolves according to
\begin{equation} \label{eq:InvDynamics}
    dq^i_t = \nu^i_t\,dt.
\end{equation}
Let us denote the price from the perspective of player $i$ by $\{s^i_t\}_{0\leqslant t \leqslant T};$ we stipulate that it is given by
\begin{equation} \label{eq:MidPriceDynamics}
    ds^i_t = \mu^i_t\,dt + \alpha^i_t\mathbb{E}\left[ \frac{1}{N}\sum_{j=1}^N \nu^j_t \Bigg| \mathcal{F}^i_t\right] \,dt + dP^i_t,
\end{equation}
where $ P^i_t \in \mathbb{M}_i,$ and $\left\{ \alpha^i_t\right\}_{0\leqslant t \leqslant T}$ is the permanent price impact parameter. 

Several remarks are in order. We observe that, even if $\mathbb{F}^i = \mathbb{F}^j,$ $i\neq j,$ we do not necessarily require that $\mu^i = \mu^j,$ or $\alpha^i = \alpha^j,$ or $P^i = P^j.$ We can interpret this as distinct beliefs between agents $i$ and $j.$ Particularly, we do not assume the existence of a fundamental price immediately perceived by every trader, but rather that each player $i \in \mathcal{N}$ negotiates according to a price $s^i,$ e.g., that they compute using the parameters they estimated. This might be a pertinent supposition once we notice that parameters such as $\mu^i$ and $\alpha^i$ must be estimated from data, and this can lead agents to assume distinct evolution to prices such as \eqref{eq:MidPriceDynamics} --- specially in short time horizons. From a mathematical viewpoint, allowing for a general heterogeneity of the parameters does not constrain the analysis, although the pertinent homogeneous parameters hypothesis can be helpful to specialize our results.

Alternatively, we can think that there exist correct, although uncertain, market parameters $\alpha$ and $\mu,$ i.e., which are $\mathbb{F}-$progressively measurable, but not necessarily $\mathbb{F}^i-$progressively measurable. In this context, it is natural to consider that trader $i$ utilizes $\alpha^i_t = \mathbb{E}\left[\alpha_t|\mathcal{F}^i_t\right]$ and $\mu^i_t = \mathbb{E}\left[\mu_t|\mathcal{F}^i_t\right],$ akin to models with latent processes, see \cite{bayraktar2018mini,bismuth2019portfolio,casgrain2019trading,firoozi2018mean}. If we were to assume this, then there is no difference in agents' beliefs having the same level of information.

We also emphasize that, for $j\neq i,$ the strategy $\nu^j$ of trader $j$ need not be $\mathbb{F}^i-$adapted; hence, the average
$$
\frac{1}{N}\sum_{j=1}^N \nu^j
$$
is not necessarily $\mathbb{F}^i-$adapted, whence we assume that player $i$ projects it in the way we describe in \eqref{eq:MidPriceDynamics}.

Since she also undergoes a temporary price impact, usually modeled to be proportional to her trading rate through a stochastic coefficient $\left\{ \kappa^i_t \right\}_{0\leqslant t \leqslant T},$ we assume that her transaction price per share is
$$
\widehat{s}^i_t = s^i_t + \kappa^i_t\nu^i_t.
$$
In this way, her cash process $c^i$ has the dynamics
\begin{equation} \label{eq:CashDynamics}
    dc^i_t = -\widehat{s}^i_t\nu^i_t\,dt.
\end{equation}

The last preference of trader $i$ that we will introduce is her (stochastic) risk aversion parameter $\lambda^i = \left\{ \lambda^i_t \right\}_{0 \leqslant t \leqslant T}.$ It will play an important role in her performance criteria; prior to describing those, we proceed to stipulate the general assumptions that we require to be valid throughout this work --- in particular, these conditions serve to ensure that the objective criteria we assign to the players are well-defined.

\textbf{Standing assumptions.} Let us maintain $i \in \mathcal{N}$ fixed. Henceforth, we consider stochastic processes $\left\{\alpha^i_t\right\}_{0\leqslant t \leqslant T},$ $\left\{\kappa^i_t\right\}_{0\leqslant t \leqslant T},$ $\left\{\mu^i_t\right\}_{0\leqslant t \leqslant T}$ and $\left\{\lambda^i_t\right\}_{0\leqslant t \leqslant T},$ all of which are $\mathbb{F}^i-$progressively measurable, satisfying the following conditions:
\begin{itemize}
    \item[\textbf{A1}] There exist positive constants $\underline{\alpha}^i,\, \underline{\kappa}^i,\, \underline{\lambda}^i, \, \overline{\alpha}^i,\, \overline{\kappa}^i$ and $\overline{\lambda}^i$ such that
    $$
    \underline{\alpha}^i \leqslant \alpha^i_t \leqslant \overline{\alpha}^i,\, \underline{\kappa}^i \leqslant \kappa^i_t \leqslant \overline{\kappa}^i_t \text{ and } \underline{\lambda}^i \leqslant \lambda^i_t \leqslant \overline{\lambda}^i,
    $$
    for $dt \times d\mathbb{P}-$a.e.a.s.;
    \item[\textbf{A2}] The drift $\left\{ \mu^i_t \right\}_{0\leqslant t \leqslant T}$ belongs to $\mathbb{L}^2;$
    \item[\textbf{A3}] The processes $\left\{ \alpha^i_t \right\}_{0\leqslant t \leqslant T}$ are semimartingales of the form
    $$
    d\alpha^i_t = \beta^i_t\,dt + dM^{\alpha^i}_t,
    $$
    where $\left\{  \beta^i_t \right\}_{0\leqslant t \leqslant T}$ is an essentially bounded process, uniformly on time, and $M^{\alpha^i} \in \mathbb{M}_i.$
\end{itemize}

The following quantities will figure in the estimates of some of this paper's main results: for each real number $u,$ we set
\begin{equation} \label{eq:KeyConstants}
    \begin{cases}
        c_1(u) := \min_{i\in\mathcal{N}}\inf_{0\leqslant t \leqslant T} \left(\frac{1}{2\kappa^i_t} - \frac{ u^2}{8\left( \kappa^i_t \right)^2} \right) \\
        \text{and } c_2(u) := \min_{i\in\mathcal{N}}\inf_{0\leqslant t \leqslant T}\left(2\lambda^i_t - \frac{1}{N}\beta^i_t - \frac{\left(\alpha^i_t\right)^2}{2}u^2 \right).
    \end{cases}
\end{equation}
We also write, from now on,
$$
\underline{\kappa} := \min_{i \in \mathcal{N}} \underline{\kappa}^i,\, \overline{\kappa} := \max_{i \in \mathcal{N}} \overline{\kappa}^i,\, \underline{\alpha} := \min_{i \in \mathcal{N}} \underline{\alpha}^i,\, \overline{\alpha} := \max_{i \in \mathcal{N}} \overline{\alpha}^i,\, \underline{\lambda} := \min_{i \in \mathcal{N}} \underline{\lambda}^i,\,\overline{\lambda} := \max_{i \in \mathcal{N}} \overline{\lambda}^i, 
$$
as well as
$$
\overline{\beta}^i := \sup_{t\in\left[0,T\right]} \esssup \left|\beta^i_t\right| \text{ and } \overline{\beta} := \max_{i \in \mathcal{N}} \overline{\beta}^i.
$$

The dynamic assumption we made on the parameter $\alpha^i$ holds if it is of the form 
$$
\alpha^i_t = f\left( \boldsymbol{Y}_t \right) \hspace{1.0cm} \left(t \in \left[0,T\right]\right),
$$
where:
\begin{itemize}
    \item For some $d \geqslant 1$, the function $f : \mathbb{R}^d \rightarrow \mathbb{R}$ is strictly positive, has essentially bounded weak derivatives up to order two (i.e., $f \in W^{2,\infty}\left(\mathbb{R}^d\right)$), and $\nabla f$ has compact support;
    \item The $d-$dimensional factor $\boldsymbol{Y}$ satisfies
    $$
    d\boldsymbol{Y}_t = \boldsymbol{\beta}^{\boldsymbol{Y}}_t\,dt + d\boldsymbol{M}^{\boldsymbol{Y}}_t;
    $$
    \item The drift $\left\{ \boldsymbol{\beta}^{\boldsymbol{Y}}_t \right\}_{0\leqslant t \leqslant T}$ is an $\mathbb{F}^i-$progressively measurable stochastic process, and $\left\{ \boldsymbol{\beta}^{\boldsymbol{Y}}_t \mathbb{I}_{ \left\{ \boldsymbol{Y}_t \in \spt\left(\nabla f\right) \right\} } \right\}_{0\leqslant t \leqslant T}$ is essentially bounded;
    \item The martingale $\boldsymbol{M}^{\boldsymbol{Y}} \in \mathbb{M}_i$ is such that $\left\{ \left\langle \boldsymbol{M}^{\boldsymbol{Y}} \right\rangle_t \mathbb{I}_{ \left\{ \boldsymbol{Y}_t \in \spt\left(\nabla f\right) \right\} } \right\}_{0\leqslant t \leqslant T}$ is essentially bounded.
\end{itemize}
Above, $\mathbb{I}_S$ denotes the indicator function of a set $S.$ A particular instance of this class is, e.g., when we take $\boldsymbol{Y}$ to be a one-dimensional Ornstein-Uhlenbeck or Cox-Ingersoll-Ross process, and $f(y) =\left( \underline{\alpha}^i \vee y \right) \wedge \overline{\alpha}^i,$ $0<\underline{\alpha}^i<\overline{\alpha}^i.$

\textbf{Optimization criteria.} Let us assume that the agent $i$ benchmarks her terminal performance by her initial wealth marked-to-market, i.e., $c^i_0 + q^i_0 s^i_0,$ and utilizes a quadratic penalty for holding inventory or ending up with it:
\begin{align}\label{eq:main_payoff}
    \begin{split}
        J_i(\nu^i;\boldsymbol{\nu}^{-i}) &= \mathbb{E}\left[c^i_T + q^i_T s^i_T - \int_0^T \lambda^i_t \left(q^i_t\right)^2\,dt - A^i\left(q^i_T\right)^2  - \left(c^i_0 + q^i_0 s^i_0 \right) \right] \\
        &= \mathbb{E}\Bigg[\int_0^T \left(-\kappa^i_t \left( \nu^i_t \right)^2 -\lambda^i_t\left(q^i_t\right)^2 + \frac{\alpha^i_t q^i_t}{N}\sum_{j= 1}^N \nu^j_t + q^i_t\mu^i_t\right) dt - A^i (q^i_T)^2 \Bigg].
    \end{split}
\end{align}
For strategy profiles $\boldsymbol{\nu} \in \mathbb{A}_{(N)},$ the functional $J_i$ is well-defined by \textbf{A1} and \textbf{A2}. The parameter $A^i > 0$ is a preference of player $i,$ and it represents her terminal inventory penalization. Regarding $\lambda^i,$ we can think that $ \lambda^i_t = \frac{1}{2}\gamma^i \left( \sigma^i_t \right)^2,$ where $\left(\sigma_t^i\right)^2$ is the price's variance from the perspective of the corresponding trader. The constant $\gamma^i$ here stands for the risk aversion level of the trader, akin to the mean-variance modeling, see \cite{almgren2001optimal} (cf. Equations (4), (5) and (15) therein), or to a constant absolute risk aversion (CARA) setting, see \cite[Chapter 6]{cartea2015algorithmic} or \cite[Chapter 3]{gueant2016financial}. 

Let us remark that, although players aim to finish with zero inventory, their initial holdings need not be positive. Thus, if trader $i$ is such that $q^i_0 \leqslant 0$ (respectively, $q^i_0 \geqslant 0$), then she is targeting to acquire (respectively, liquidate) $\left|q^i_0\right|$ shares of the asset. If an agent begins with $q^i_0 \equiv 0,$ then she will carry out an arbitrage program.

\textbf{Nash Equilibria} Our objective is to investigate Nash equilibria determined by the set of functionals $\{J_1,...,J_N\}$ in terms of the following definition.

\begin{definition}[The unconstrained setting]\label{def:nasheq}
A set of admissible strategy profiles ${\boldsymbol{\nu}}^*=\left(\nu^{*1}, \cdots, \nu^{*N}\right)^\intercal \in \mathbb{A}_{(N)}$ is a Nash equilibrium for the unconstrained game if, for each $i \in \mathcal{N},$
\begin{equation}\label{eq:nasheq}
\nu^{*i} = \argmax_{\nu^i \in \mathbb{A}_i} J_{i}\left( \nu^{i}; {\boldsymbol{\nu}}^{*-i}\right).
\end{equation}
\end{definition}

We will also consider equilibria in the constrained setting, i.e., in the framework in which all players demand full execution by terminal time. We introduce the constrained admissible control set
$$
\mathbb{A}_{c,i} := \left\{ \nu^i \in \mathbb{A}_{i} : \int_0^T \nu^i_t\,dt = -q^i_0,\, \mathbb{P}-\text{a.s.} \right\} \hspace{1.0cm} (i\in\mathcal{N}).
$$
Moreover, we set
$$
\mathbb{A}_{c,-i} := \Pi_{j\neq i}\mathbb{A}_{c,j},
$$
$$
\mathbb{A}_{c,i}^0 := \left\{ \nu^i \in \mathbb{A}_{i} : \int_0^T \nu^i_t\,dt = 0,\, \mathbb{P}-\text{a.s.} \right\},
$$
for $i \in \mathcal{N},$ and
$$
\mathbb{A}_{c,(N)} := \Pi_{j=1}^N \mathbb{A}_{c,j},\, \mathbb{S}_{c,(N)} := \mathbb{S}_{(N)} \cap \mathbb{A}_{c,(N)},\, \mathcal{M}_{(N)} := \Pi_{i=1} \mathcal{M}_i,
$$
as well as
$$
\mathbb{A}_{c,(N)}^0 := \Pi_{j=1}^N \mathbb{A}_{c,j}^0.
$$

\begin{definition}[The constrained setting]\label{def:nasheq_Constrained}
A Nash equilibrium for the constrained problem is a stochastic process ${\boldsymbol{\nu}}^*=\left(\nu^{*1}, \cdots, \nu^{*N}\right)^\intercal \in \mathbb{A}_{c,(N)}$ such that
\begin{equation}\label{eq:nasheq_Constrained}
\nu^{*i} = \argmax_{\nu^i \in \mathbb{A}_{c,i}} J_{i}\left( \nu^{i}; {\boldsymbol{\nu}}^{*-i}\right),
\end{equation}
for each $i \in \mathcal{N}.$
\end{definition}

\section{Analysis of the \texorpdfstring{$N-$}{N-}player game: the unconstrained setting} \label{sec:NplayerGame}

\subsection{The general unconstrained setting}

We base our approach here on the variational formulation. It allows us to characterize the speeds of trading comprising the Nash equilibrium, alongside their corresponding inventories, as the solution of an FBSDE system. We develop this below. Our starting point is a lemma.
\begin{lemma} \label{lem:concavity}
    Let us assume the following conditions
    $$
    \underline{\lambda}^i > \frac{\overline{\beta}^i}{2N} \text{ and } A^i > \frac{\overline{\alpha}^i}{2N},
    $$
    for every $i\in \mathcal{N}.$ Then, given $i\in \mathcal{N}$ and $\boldsymbol{\nu}^{-i} \in \mathbb{A}_{-i},$ the functional $w^i \mapsto J_i(w^i;\boldsymbol{\nu}^{-i})$ is strictly concave. 
\end{lemma}
\begin{proof}
Let us fix $w^i, \widetilde{w}^i \in \mathbb{A}_i,$ $\boldsymbol{\nu}^{-i} \in \mathbb{A}_{-i},$ $0 \leqslant \theta \leqslant 1.$ Denote by $q^i$ and $\widetilde{q}^i$ the corresponding inventory processes associated with $w^i$ and $\widetilde{w}^i,$ respectively. We have
\begin{align*}
    J_i&( \theta w^i + (1-\theta)\widetilde{w}^i; \boldsymbol{\nu}^{-i}) -\theta J_i(w^i; \boldsymbol{\nu}^{-i}) - (1-\theta) J_i(\widetilde{w}^i; \boldsymbol{\nu}^{-i}) \\
    &= \theta(1-\theta)\mathbb{E}\bigg[\int_0^T \bigg\{\kappa^i_t \left(w^i_t - \widetilde{w}^i_t \right)^2 + \lambda^i_t \left(q^i_t - \widetilde{q}^i_t \right)^2  - \frac{\alpha^i_t}{N}\left( q^i_t - \widetilde{q}^i_t \right)\left( w^i_t - \widetilde{w}^i_t \right) \bigg\}\,dt + A^i\left( q^i_T - \widetilde{q}^i_T \right)^2 \bigg] \\
    &= \theta(1-\theta)\mathbb{E}\bigg[ \int_0^T \left\{ \kappa_t^i \left(w^i_t - \widetilde{w}^i_t \right)^2 + \left( \lambda_t^i + \frac{\beta_t^i}{2N} \right) \left(q^i_t - \widetilde{q}^i_t \right)^2\right\} dt + \left( A^i - \frac{\alpha^i_T}{2N}\right)\left( q^i_T - \widetilde{q}^i_T \right)^2 \bigg] \\
    &\geqslant 0,
\end{align*}
with equality holding above if, and only if, $w^i = \widetilde{w}^i.$ This argument shows the strict concavity of $w^i \mapsto J_i(w^i, \boldsymbol{\nu}^{-i}).$ 
 
\end{proof}


The following definition will be of great technical importance from now on.
\begin{definition}
Given $i\in\mathcal{N},$ $\boldsymbol{\nu} \in \mathbb{A}_{(N)},$ and $w^i \in \mathbb{A}_i,$ the $i-$th partial G\^ateaux derivative of $J_i$ in the point $\boldsymbol{\nu},$ in the direction $w^i,$ is defined as
\begin{equation} \label{eq:GateauxDerivativeDefn}
    \left\langle D_i J_i(\nu^i; \boldsymbol{\nu}^{-i}) , w^i \right\rangle := \lim_{\epsilon \rightarrow 0} \frac{J_i(\nu^i + \epsilon w^i, \boldsymbol{\nu}^{-i}) - J_i(\nu^i, \boldsymbol{\nu}^{-i})}{\epsilon}.
\end{equation}
\end{definition}

\begin{lemma} \label{lem:GateauxDerivWellDefnd}
The G\^ateaux derivative \eqref{eq:GateauxDerivativeDefn} is well-defined for each $i\in \mathcal{N},$ $\boldsymbol{\nu}\in\mathbb{A}_{(N)},$ and $w^i\in \mathbb{A}_i.$
\end{lemma}
\begin{proof}
Let us write
$$
\frac{J_i(\nu^i + \epsilon w^i, \boldsymbol{\nu}^{-i}) - J_i(\nu^i, \boldsymbol{\nu}^{-i})}{\epsilon} = I_1^\epsilon + I_2^\epsilon + I_3^\epsilon + I_4^\epsilon + I_5^\epsilon +  I_6^\epsilon,
$$
where
$$
\begin{cases}
I_1^\epsilon := -\mathbb{E}\left[ \int_0^T \kappa^i_t\left\{ \frac{\left( \nu^i + \epsilon w^i_t \right)^2 - \left( \nu^i_t \right)^2}{\epsilon} \right\}\,dt\right],\\
I_2^\epsilon := -\mathbb{E}\left[ \int_0^T \lambda^i_t\left\{ \frac{\left( q^i + \epsilon \int_0^t w^i_u\,du \right)^2 - \left( q^i_t \right)^2}{\epsilon} \right\}\,dt\right],\\
I_3^\epsilon := \frac{1}{N\epsilon}\mathbb{E}\left[\int_0^T \alpha^i_t \left\{  \left(q^i_t + \epsilon\int_0^t w^i_u\,du \right)\left( \nu^i_t + \epsilon w^i_t\right) - q^i_t \nu^i_t \right\}\,dt\right],\\
I_4^\epsilon := \frac{1}{N\epsilon}\mathbb{E}\left[\int_0^T \alpha^i_t\left\{\left(q^i_t + \epsilon\int_0^t w^i_u\,du \right)\sum_{j\neq i}\nu^j_t - q^i_t\sum_{j\neq i} \nu^j_t \right\}\,dt\right],\\
I_5^\epsilon := \frac{1}{\epsilon} \mathbb{E}\left[ \int_0^T \left\{ \left(q^i_t + \epsilon \int_0^t w^i_u\,du \right)\mu^i_t - q^i_t \mu^i_t \right\}\,dt\right], \\
I_6^\epsilon := -\mathbb{E}\left[ A^i\left\{ \frac{\left( q^i_T + \epsilon\int_0^Tw^i_t\,dt\right)^2 - \left(q^i_T\right)^2 }{\epsilon} \right\} \right] .
\end{cases}
$$
We notice that
\begin{equation} \label{eq:I1eps}
    I_1^\epsilon = -\mathbb{E}\left[\int_0^T 2\kappa^i_t w^i_t\nu^i_t \,dt \right] + \epsilon\|w^i\|_0^2 \xrightarrow{\epsilon \rightarrow 0} -\mathbb{E}\left[\int_0^T 2\kappa^i_t w^i_t\nu^i_t \,dt \right].
\end{equation}
Similarly, since $\left\{ \int_0^t w^i_u\,du \right\}_{0\leqslant t \leqslant T} \in \mathbb{L}^2,$ we prove that
\begin{equation} \label{eq:I2eps}
    \lim_{\epsilon \rightarrow 0} I_2^\epsilon = -\mathbb{E}\left[ \int_0^T 2\lambda_t q^i_t \int_0^t w^i_u\,du \,dt \right],
\end{equation}
\begin{equation} \label{eq:I3eps}
    \lim_{\epsilon \rightarrow 0} I_3^\epsilon = \mathbb{E}\left[ \int_0^T \frac{\alpha^i_t}{N}\left(\nu^i_t \int_0^t w^i_u\,du + q^i_t w^i_t \right)\,dt \right],
\end{equation}
\begin{equation} \label{eq:I4eps}
    \lim_{\epsilon \rightarrow 0} I_4^\epsilon = \mathbb{E}\left[\int_0^T\frac{\alpha^i_t}{N}\sum_{j\neq i}\nu^j_t \int_0^t w^i_u\,du\,dt \right],
\end{equation}
\begin{equation} \label{eq:I5eps}
    \lim_{\epsilon \rightarrow 0} I_5^\epsilon = \mathbb{E}\left[\int_0^T\mu^i_t \int_0^t w^i_u\,du\,dt \right],
\end{equation}
and
\begin{equation} \label{eq:I6eps}
    \lim_{\epsilon \rightarrow 0} I_6^\epsilon = \mathbb{E}\left[ -2A^i q^i_T \int_0^T w^i_t\,dt \right].
\end{equation}
Gathering \eqref{eq:I1eps}-\eqref{eq:I5eps} together, and integrating by parts when necessary, we deduce that the limit in the right-hand side of \eqref{eq:GateauxDerivativeDefn} exists and is equal to
\begin{align} \label{eq:GateauxDerivativeFormula}
    \begin{split}
        \left\langle D_iJ_i\left(\nu^i;\boldsymbol{\nu}^{-i}\right), w^i\right\rangle = \mathbb{E}\Bigg[\int_0^T w^i_t\Bigg\{&-2\kappa^i_t \nu^i_t -2 \int_t^T \lambda^i_u q^i_u\,du + \frac{\alpha^i_t}{N}q^i_t \\
         &+ \int_t^T \frac{\alpha^i_u}{N}\sum_{j=1}^N  \nu^j_u\,du + \int_t^T \mu^i_u\,du -2A^iq^i_T\Bigg\}\,dt \Bigg].
    \end{split}
\end{align}  
\end{proof}

\begin{corollary} \label{cor:PartialOptmFBSDE}
    Let us suppose that the assumptions of Lemma \ref{lem:concavity} hold. We consider $i\in\mathcal{N}$ and $\boldsymbol{\nu}^{ -i} \in \mathbb{A}_{-i}.$ A strategy $\nu^{*i} \in \mathbb{A}_i$ solves the optimization problem
    \begin{equation} \label{eq:MaximizeIthDirection}
        \nu^{*i} = \argmax_{\nu^i \in \mathbb{A}_i} J_i\left( \nu^i; \boldsymbol{\nu}^{-i}\right)
    \end{equation}
    if, and only if, 
    \begin{equation} \label{eq:GateauxDerivativeIsZero}
        \left\langle D_i J_i\left( \nu^{*i};\boldsymbol{\nu}^{-i} \right) , w^i \right\rangle = 0,
    \end{equation}
    for all $w^i \in \mathbb{A}_i.$ Consequently, $\nu^{*i}$ must solve
    \begin{equation} \label{eq:IthFBSDE}
        \begin{cases}
        2\kappa^i_t\nu^{*i}_t = x^{*i}_t,\\
        dq^{*i}_t = \frac{x^{*i}_t}{2\kappa^i_t}\,dt,\\
        -dx^{*i}_t = -2\lambda^i_t q^{*i}_t\,dt + \mathbb{E}\left[\frac{\alpha^i_t}{N}\sum_{j\neq i}  \nu^j_t\Big|\mathcal{F}^i_t\right]\,dt -\frac{q^{*i}_t}{N}\beta^i_t\,dt + \mu^i_t\,dt - dM^{*i}_t,\\
        -x_T^{*i} = \left(2A^i - \frac{\alpha^i_T}{N} \right)q^{*i}_T,
        \end{cases}
    \end{equation}
    for some $ M^{*i} \in \mathbb{M}_i.$
\end{corollary}
\begin{proof}
The first part is standard, see \cite[Theorem 1.3]{lions1971optimal} or item ($a$) of the proof of \cite[Chapter II, Proposition 2.1]{ekeland1999convex}. To demonstrate the other half, we rewrite \eqref{eq:GateauxDerivativeFormula}, with the aid of the tower property of conditional expectations, in the following manner:
\begin{align} \label{eq:GateauxDerivativeFormula2}
    \begin{split}
        \langle D_iJ_i\left(\nu^i;\boldsymbol{\nu}^{-i}\right), w^i\rangle = \mathbb{E}\Bigg[\int_0^T w^i_t\Bigg\{&-x^i_t + 2 \int_0^t \lambda^i_u q^i_u\,du + \frac{\alpha^i_t}{N}q^i_t - \frac{1}{N}\int_0^t \frac{\alpha^i_u}{2\kappa^i_u} x^i_u\,du  \\
        & - \int_0^t \mathbb{E}\left[\frac{\alpha^i_u}{N}\sum_{j\neq i}  \nu^j_u \Bigg|\mathcal{F}^i_u\right]\,du - \int_0^t \mu^i_u\,du + M^i_t\Bigg\}\,dt \Bigg].
    \end{split}
\end{align}
where $x^i = 2\kappa^i_t\nu^i_t$ and
\begin{align*}
    \begin{split}
        M^i_t :=&\, \mathbb{E}\left[ \int_0^T \left( - 2\lambda^i_u q^i_u + \frac{1}{N} \frac{\alpha^i_u}{2\kappa^i_u} x^i_u + \frac{\alpha^i_u}{N}\sum_{j\neq i}  \nu^j_u + \mu^i_u\right)\,du -2A^i q^i_T \Bigg| \mathcal{F}^i_t\right] \\
        &- \int_0^t \left\{  \mathbb{E}\left[\frac{\alpha^i_u}{N}\sum_{j\neq i}  \nu^j_u\Bigg|\mathcal{F}^i_t\right] - \mathbb{E}\left[\frac{\alpha^i_u}{N}\sum_{j\neq i}  \nu^j_u\Bigg|\mathcal{F}^i_u\right] \right\}\,du.
    \end{split}
\end{align*}
We emphasize that $ M^i $ belongs to $\mathbb{M}_i.$ We observe that \eqref{eq:GateauxDerivativeIsZero} must be valid for every $w^i \in \mathbb{L}^2$ when $\nu^i = \nu^{*i}.$ Therefore, we conclude that $\nu^{*i}$ solves \eqref{eq:MaximizeIthDirection} if, and only if, there exists $ M^{*i} \in \mathbb{M}_i$ such that $(q^{*i},x^{*i} = 2\kappa^i\nu^{*i},M^{*i})$ is a solution of the FBSDE \eqref{eq:IthFBSDE}.
 
\end{proof}

From Corollary \ref{cor:PartialOptmFBSDE}, we obtain the subsequent characterization of Nash equilibria in the present context.
\begin{corollary} \label{cor:CharactNashHeterogReg}
    Under the assumptions of Lemma \ref{lem:concavity}, a strategy $\boldsymbol{\nu}^*$ is a Nash equilibrium if, and only if, the processes $(\boldsymbol{q}^*,\boldsymbol{\nu}^*,\boldsymbol{M}^*) \in \mathbb{S}_{(N)}\times\mathbb{A}_{(N)}\times \mathbb{M}_{(N)} $ solve the FBSDE
    \begin{equation} \label{eq:FBSDEHeterog}
        \begin{cases}
        \boldsymbol{x}^*_t = \boldsymbol{K}_t\boldsymbol{\nu}_t^*,\\
        d\boldsymbol{q}^*_t = \boldsymbol{K}_t^{-1}\boldsymbol{x}^*_t\,dt,\\
        -d\boldsymbol{x}^*_t = \boldsymbol{\mathcal{P}}_t\left( \boldsymbol{C}_t\boldsymbol{x}^*_t\right)\,dt - \boldsymbol{\Sigma}_t \boldsymbol{q}^*_t\,dt + \frac{1}{N}\boldsymbol{\beta}_t\boldsymbol{q}^*_t\,dt + \boldsymbol{\mu}_t\,dt - d\boldsymbol{M}^*_t,\\
        \boldsymbol{q}^*_0 = \boldsymbol{q}_0 \text{ and } -\boldsymbol{x}^*_T = \boldsymbol{D}_T\boldsymbol{q}^*_T,
        \end{cases}
    \end{equation}
    where $\boldsymbol{M}^* \in \mathbb{M}_{(N)}$ and the stochastic matrix coefficients are given by
    $$
    \begin{cases}
    \boldsymbol{K}_t := \diag\left(2\kappa^i_t\right),\\
    \boldsymbol{C}_t := \left( (1-\delta_{ij})\frac{\alpha^i_t}{2N\kappa^j_t} \right)_{ij},\\
    \boldsymbol{\Sigma}_t := \diag\left( 2\lambda^i_t \right),\\
    \boldsymbol{\beta}_t := \diag\left( \beta^i_t \right),\\
    \boldsymbol{\mu}_t := \left( \mu^1_t,...,\mu^N \right)^\intercal,\\ 
    \text{and } \boldsymbol{D}_t := \diag\left( 2A^i - \frac{\alpha^i_t}{N}\right).
    \end{cases}
    $$
\end{corollary}

\begin{theorem} \label{thm:NashEquilibriumGeneralSetting}
Let us assume that the model parameters satisfy
$$
\overline{\alpha}^2<16\underline{\kappa}\left(\underline{\lambda} - \frac{1}{2N}\overline{\beta}\right)
$$
and
$$
\underline{D} := \min_{i \in \mathcal{N} }\left(2A^i - \frac{\overline{\alpha}^i}{N} \right) > 0.
$$
Then, the FBSDE \eqref{eq:FBSDEHeterog} admits a unique solution $\left(\boldsymbol{q}^*,\boldsymbol{\nu}^*,\boldsymbol{M}^*\right) \in \mathbb{S}_{(N)} \times \mathbb{S}_{(N)} \times  \mathbb{M}_{(N)}$ (or, equivalently, $\left(\boldsymbol{q}^*,\boldsymbol{x}^*,\boldsymbol{M}^*\right) \in \mathbb{S}_{(N)} \times \mathbb{S}_{(N)} \times  \mathbb{M}_{(N)}).$
\end{theorem}
\begin{remark} \label{rem:KeyRemark}
If a constant $\theta$ satisfies
$$
\frac{\overline{\alpha}^2}{4\left(\underline{\lambda} - \frac{1}{2N}\overline{\beta}\right)} < \theta < 4\underline{\kappa},
$$
then taking $a := \sqrt{\theta}$ yields $c_1:=c_1(a)>0$ and $c_2:= c_2(1/a) > 0$ (c.f. \eqref{eq:KeyConstants}).
\end{remark}
\begin{proof}
We will demonstrate this Theorem with a continuation method developed in \cite{peng1999fully}. Our approach is similar to that of \cite[Theorem 4.2]{fujii2020mean}.

Let us consider the set $I$ of all $\rho \in \left[0,1\right[$ for which the FBSDE
    \begin{equation} \label{eq:FBSDEConnectivity}
        \begin{cases}
        d\boldsymbol{q}^\rho_t = \rho\boldsymbol{K}_t^{-1}\boldsymbol{x}^\rho_t\,dt + \boldsymbol{f}_t\,dt,\\
        -d\boldsymbol{x}^\rho_t = -\left(1-\rho\right) \boldsymbol{q}^\rho_t + \rho\left(\boldsymbol{\mathcal{P}}_t\left( \boldsymbol{C}_t\boldsymbol{x}^\rho_t \right) - \boldsymbol{\Sigma}_t \boldsymbol{q}^\rho_t + \frac{1}{N}\boldsymbol{\beta}_t\boldsymbol{q}^\rho_t \right)\,dt  +\, \rho\boldsymbol{\mu}_t + \boldsymbol{g}_t\,dt - d\boldsymbol{M}^\rho_t,\\
        \boldsymbol{q}^\rho_0 = \boldsymbol{q}_0 \text{ and } -\boldsymbol{x}^\rho_T = \left( 1-\rho \right) \boldsymbol{q}^\rho_T + \rho\boldsymbol{D}_T\boldsymbol{q}^\rho_T + \boldsymbol{\eta},
        \end{cases}
    \end{equation}
has a unique solution $\left( \boldsymbol{q}^\rho, \boldsymbol{x}^\rho, \boldsymbol{M}^\rho\right)$ with continuous paths, for any given $\boldsymbol{f},\boldsymbol{g} \in \mathbb{A}_{(N)}$ and $\boldsymbol{\eta} \in \Pi_{i=1}L^2(\Omega,\mathcal{F}^i_T).$ It is immediate to verify that $0\in I.$

We assume $\rho \in I$ and prove that $\rho+\zeta$ will still belong to $I$ for sufficiently small $\zeta> 0.$ Indeed, for each $\left(\boldsymbol{q},\boldsymbol{x}\right),$ the current assumptions guarantee the existence of the solution $\left( \boldsymbol{X},\boldsymbol{Q},\boldsymbol{M}\right)$ of the FBSDE: 
    \begin{equation} \label{eq:FBSDEPfStep1}
        \begin{cases}
        d\boldsymbol{Q}_t = \rho\boldsymbol{K}_t^{-1}\boldsymbol{X}_t\,dt + \zeta\boldsymbol{K}_t^{-1}\boldsymbol{x}_t\,dt + \boldsymbol{f}_t\,dt,\\
        -d\boldsymbol{X}_t = -\left(1-\rho\right) \boldsymbol{Q}_t\,dt + \rho\left(\boldsymbol{\mathcal{P}}_t\left( \boldsymbol{C}_t\boldsymbol{X}_t\right) - \boldsymbol{\Sigma}_t \boldsymbol{Q}_t + \frac{1}{N}\boldsymbol{\beta}_t\boldsymbol{Q}_t \right)\,dt\\
        \hspace{1.4cm} +\, \zeta \boldsymbol{q}_t + \zeta\left(\boldsymbol{\mathcal{P}}_t\left(\boldsymbol{C}_t\boldsymbol{x}_t\right) - \boldsymbol{\Sigma}_t \boldsymbol{q}_t + \frac{1}{N}\boldsymbol{\beta}_t\boldsymbol{q}_t \right)\,dt \\
        \hspace{1.4cm} +\, \left(\rho + \zeta\right)\boldsymbol{\mu}_t + \boldsymbol{g}_t\,dt - d\boldsymbol{M}_t,\\
        \boldsymbol{Q} = \boldsymbol{q}_0 \text{ and } -\boldsymbol{X}_T = \left( 1-\rho \right) \boldsymbol{Q}_T + \rho\boldsymbol{D}_T\boldsymbol{Q}_T + \zeta\boldsymbol{q}_T + \zeta\boldsymbol{D}_T\boldsymbol{q}_T + \boldsymbol{\eta}.
        \end{cases}
    \end{equation}
We will prove that the mapping $\left( \boldsymbol{q},\boldsymbol{x} \right) \mapsto \left( \boldsymbol{Q}, \boldsymbol{X} \right)$ is a contraction, as long as $\zeta > 0$ is sufficiently small. In effect, let $\left(\boldsymbol{q},\boldsymbol{x}\right)$ and $\left(\boldsymbol{q}',\boldsymbol{x}'\right)$ correspond to solutions $\left( \boldsymbol{Q},\boldsymbol{X},\boldsymbol{M} \right)$ and $\left( \boldsymbol{Q}^\prime,\boldsymbol{X}^\prime,\boldsymbol{M}^\prime \right),$ respectively. We write
$$
\begin{cases}
\Delta\boldsymbol{Q} := \boldsymbol{Q} - \boldsymbol{Q}^\prime,\\
\Delta\boldsymbol{X} := \boldsymbol{X} - \boldsymbol{X}^\prime,\\
\text{ and } \Delta\boldsymbol{M} := \boldsymbol{M} - \boldsymbol{M}^\prime.
\end{cases}
$$
On the one hand, using the It\^o's product formula, we infer
\begin{equation} \label{eq:CrossExpec}
    \mathbb{E}\left[ \Delta\boldsymbol{Q}_T \cdot \Delta\boldsymbol{X}_T \right] = I_1 + I_2,
\end{equation}
where
\begin{align}
    \begin{split}
        I_1 := \mathbb{E}\Bigg[\int_0^T\bigg\{& \rho \Delta \boldsymbol{X}_t \cdot \boldsymbol{K}_t^{-1} \Delta\boldsymbol{X}_t - \rho \Delta \boldsymbol{Q}_t \cdot \boldsymbol{\mathcal{P}}_t\left( \boldsymbol{C}_t\Delta \boldsymbol{X}_t \right) \\
        & +\Delta\boldsymbol{Q}_t \cdot \left[ \left(1 -\rho\right)\boldsymbol{I} + \rho\left(\boldsymbol{\Sigma}_t - \frac{1}{N}\boldsymbol{\beta}_t\right) \right]\Delta \boldsymbol{Q}_t \bigg\}\,dt \Bigg],
    \end{split}
\end{align}
and
\begin{align}
    \begin{split}
        I_2 := \zeta\mathbb{E}\Bigg[\int_0^T \bigg\{\Delta\boldsymbol{X}_t\cdot \boldsymbol{K}_t^{-1} \Delta \boldsymbol{x}_t  + \Delta\boldsymbol{Q}_t \cdot \left( \boldsymbol{I} + \boldsymbol{\Sigma}_t - \frac{1}{N}\boldsymbol{\beta}_t \right)\Delta\boldsymbol{q}_t  + \Delta\boldsymbol{Q}_t \cdot \boldsymbol{\mathcal{P}}_t\left( \boldsymbol{C}_t\Delta\boldsymbol{x}_t\right) \bigg\}\,dt \Bigg].
    \end{split}
\end{align}

With the aid of Young's inequality, the conditional version Jensen's inequality, and the tower property of conditional expectations, we obtain
\begin{equation}
    |\mathbb{E} [ \Delta \boldsymbol{Q}_t \cdot \boldsymbol{\mathcal{P}}_t\left( \boldsymbol{C}_t\Delta \boldsymbol{X}_t \right) ] | \leqslant \mathbb{E}\left[ \frac{a^2}{2}\sum_{i=1}^N \frac{\left(\Delta X^i_t\right)^2}{\left(2\kappa^i_t\right)^2}  + \frac{1}{2 a^2}\sum_{i=1}^N\left(\alpha^i_t\right)^2 \left(\Delta Q^i_t\right)^2 \right],
\end{equation}
where we have written 
$$
\Delta\boldsymbol{Q}_t = \left(\Delta Q^1_t, ..., \Delta Q^N_t \right)^\intercal \text{ and } \Delta\boldsymbol{X}_t = \left(\Delta X^1_t, ..., \Delta X^N_t \right)^\intercal.
$$
Let us fix $a, c_1$ and $c_2$ as described in Remark \ref{rem:KeyRemark}. Therefore, we estimate
\begin{align} \label{eq:I1Estimate}
    \begin{split}
        I_1 &\geqslant \rho c_1\|\Delta \boldsymbol{X} \|^2 + \left[ \left(1-\rho\right) + \rho c_2 \right]\|\Delta \boldsymbol{Q}\|^2 \\
        &\geqslant \left( 1 \wedge c_2 \right) \left\|\Delta \boldsymbol{Q}\right\|^2,
    \end{split}
\end{align}
and also
\begin{align} \label{eq:I2Estimate}
    \left| I_2 \right| \leqslant C\zeta \mathbb{E}\left[\int_0^T \left\{ \left|\Delta\boldsymbol{X}_t\right| \left|\Delta \boldsymbol{x}_t\right| + \left|\Delta\boldsymbol{Q}_t\right|\left( \left|\Delta \boldsymbol{x}_t\right| + \left|\Delta \boldsymbol{q}_t\right| \right) \right\}\,dt\right].
\end{align}
Altogether, from \eqref{eq:CrossExpec}, \eqref{eq:I1Estimate} and \eqref{eq:I2Estimate} we deduce
\begin{align} \label{eq:ContinuationPart1}
    \begin{split}
        \mathbb{E} [ \Delta\boldsymbol{Q}_T \cdot \Delta\boldsymbol{X}_T ] \geqslant \gamma \left\|\Delta \boldsymbol{Q}\right\|^2 - C\zeta \mathbb{E}\left[\int_0^T \left\{ \left|\Delta\boldsymbol{X}_t\right| \left|\Delta \boldsymbol{x}_t\right| + \left|\Delta\boldsymbol{Q}_t\right|\left( \left|\Delta \boldsymbol{x}_t\right| + \left|\Delta \boldsymbol{q}_t\right| \right) \right\}\,dt\right].
    \end{split}
\end{align}

On the other hand, the terminal condition of $\Delta\boldsymbol{X}$ gives
\begin{align} \label{eq:ContinuationPart2}
    \begin{split}
        \mathbb{E}[ \Delta\boldsymbol{Q}_T \cdot \Delta\boldsymbol{X}_T ] &= -\mathbb{E}[\Delta\boldsymbol{Q}_T\cdot [\left( 1-\rho\right)\Delta\boldsymbol{Q}_T+\rho\boldsymbol{D}_T\Delta\boldsymbol{Q}_T +\zeta\Delta\boldsymbol{q}_T +\zeta\boldsymbol{D}_T\Delta\boldsymbol{q}_T ] ] \\
        &\leqslant - c_0\|\Delta\boldsymbol{Q}_T\|_{L^2\left(\Omega,\mathcal{F}\right)^N}^2 + C\zeta \mathbb{E}\left[\left|\Delta\boldsymbol{Q}_T\right|\left|\Delta\boldsymbol{q}_T\right|\right],
    \end{split}
\end{align}
with $c_0 := 1 \wedge \underline{D}.$ Using Young's inequality, and assuming $\zeta$ to be sufficiently small, \eqref{eq:ContinuationPart1} and \eqref{eq:ContinuationPart2} yield
\begin{align}  \label{eq:CoreEstimate}
    \begin{split}
        \|\Delta \boldsymbol{Q}_T\|_{L^2\left(\Omega,\mathcal{F}\right)^N}^2 + \|\Delta\boldsymbol{Q}\|^2 \leqslant\, C\zeta\|\Delta \boldsymbol{X}\|^2 + C\zeta\mathbb{E}\left[\left|\Delta \boldsymbol{q}_T\right|^2 + \int_0^T \left\{\left|\Delta \boldsymbol{q}_t \right|^2 + \left| \Delta\boldsymbol{x}_t \right|^2 \right\}\,dt \right].
    \end{split}
\end{align}
Regarding $\Delta \boldsymbol{Q}, \Delta\boldsymbol{q}$ and $\Delta \boldsymbol{x}$ as inputs in the BSDE solved by $\Delta\boldsymbol{X},$ standard stability techniques (such as, e.g., those developed in \cite[Theorem 10.5]{touzi2012optimal} or \cite[Proposition 2.2]{bouchard:hal-01158912}), together with basic properties of conditional expectations, allow us to infer
\begin{align*}
    \begin{split}
        \|\Delta\boldsymbol{X}\|_{\mathbb{S}_{(N)}}^2 + \| \Delta \boldsymbol{M} \|_{\mathbb{M}_{(N)}}^2 \leqslant&\, C\left(\|\Delta \boldsymbol{Q}_T\|_{L^2\left(\Omega,\mathcal{F}\right)^N}^2 + \|\Delta\boldsymbol{Q}\|^2\right) \\
        &+ C\zeta\left(\|\Delta \boldsymbol{q}_T\|_{L^2\left(\Omega,\mathcal{F}\right)^N}^2 + \|\Delta\boldsymbol{q}\|^2 + \|\Delta\boldsymbol{x}\|^2 \right) \\
        \leqslant&\, C\zeta\|\Delta \boldsymbol{X}\|^2+  C\zeta\left(\|\Delta \boldsymbol{q}_T\|_{L^2\left(\Omega,\mathcal{F}\right)^N}^2 + \|\Delta\boldsymbol{q}\|^2 + \|\Delta\boldsymbol{x}\|^2 \right) \\
        \leqslant&\, C\zeta\|\Delta \boldsymbol{X}\|_{\mathbb{S}_{(N)}}^2+  C\zeta\left(\|\Delta \boldsymbol{q}_T\|_{L^2\left(\Omega,\mathcal{F}\right)^N}^2 + \|\Delta\boldsymbol{q}\|^2 + \|\Delta\boldsymbol{x}\|^2 \right),
    \end{split}
\end{align*}
where we utilized \eqref{eq:CoreEstimate} in the last inequality above. Assuming $C\zeta < 1/2,$ it follows that
\begin{align} \label{eq:FinalEstX}
    \begin{split}
        \|\Delta\boldsymbol{X}\|_{\mathbb{S}_{(N)}}^2 + \| \Delta \boldsymbol{M} \|_{\mathbb{M}_{(N)}}^2 &\leqslant C\zeta\left(\|\Delta \boldsymbol{q}_T\|_{L^2\left(\Omega,\mathcal{F}\right)^N}^2 + \|\Delta\boldsymbol{q}\|^2 + \|\Delta\boldsymbol{x}\|^2 \right).
    \end{split}
\end{align}
Employing in \eqref{eq:FinalEstX} standard ODE estimates for $\Delta\boldsymbol{Q},$ we likewise obtain
\begin{equation} \label{eq:FinalEstQ}
    \|\Delta\boldsymbol{Q}\|_{\mathbb{S}_{(N)}}^2 \leqslant C\zeta\left(\|\Delta \boldsymbol{q}_T\|_{L^2\left(\Omega,\mathcal{F}\right)^N}^2 + \|\Delta\boldsymbol{q}\|^2 + \|\Delta\boldsymbol{x}\|^2 \right).
\end{equation}
Therefore, we conclude from \eqref{eq:FinalEstX} and \eqref{eq:FinalEstQ} that the mapping $\left(\boldsymbol{q},\boldsymbol{x}\right) \mapsto \left(\boldsymbol{Q},\boldsymbol{X}\right)$ is a contraction on the space $\mathbb{S}_{(N)} \times \mathbb{S}_{(N)},$ as long as $C\zeta < 1.$

We remark that, in the argument above, for a given $\rho \in I,$ we only needed $\zeta > 0$ to satisfy $\rho + \zeta \leqslant 1$ and $C\zeta < 1,$ for a certain constant $C$ depending solely on model parameters (and not on $\rho$). Consequently, $\sup I = 1.$ Let us take $\rho^* \in I$ with $1-C^{-1} < \rho^* < 1.$ Applying the above argument to $\rho=\rho^*,$ we can take $\zeta := 1-\rho^* < C^{-1}$ and infer that there exists a solution for $\rho = \rho^* + \zeta = 1$ in place of $\rho,$ finishing the proof.  
\end{proof}

\subsection{Constant model parameters} \label{subsec:constModelParams}

Throughout this subsection, let us assume that all model parameters are constant. In particular, $\boldsymbol{\beta} \equiv \boldsymbol{0}.$ We introduce the matrix $\boldsymbol{B}$ as follows:
$$
\boldsymbol{B} := \begin{bmatrix}
\boldsymbol{0}^{N\times N} & \boldsymbol{K}^{-1} \\
\boldsymbol{\Sigma} & - \boldsymbol{C}
\end{bmatrix}.
$$
\begin{lemma}
Under the assumptions of Theorem \ref{thm:NashEquilibriumGeneralSetting}, the Riccati ODE
\begin{equation} \label{eq:RiccatiODE}
    \begin{cases}
        \dot{\boldsymbol{G}} = \boldsymbol{\Sigma} -\boldsymbol{C}\boldsymbol{G} - \boldsymbol{G}\boldsymbol{K}^{-1}\boldsymbol{G},\\
        \boldsymbol{G}(T) = -\boldsymbol{D},
    \end{cases}
\end{equation}
admits a solution $\boldsymbol{G} : \left[0,T\right] \rightarrow \mathbb{R}^{N\times N},$ which is given by 
\begin{equation} \label{eq:FormulaIacopoRiccati}
    \boldsymbol{G}(t) = \boldsymbol{Y}_2(t)\boldsymbol{Y}_1(t)^{-1},
\end{equation}
where $\boldsymbol{Y}(t) = \left( \boldsymbol{Y}_1(t), \boldsymbol{Y}_2(t) \right)^\intercal \in \mathbb{R}^{2N\times N}$ is defined as 
$$
\boldsymbol{Y}(t) = e^{-(T-t)\boldsymbol{B}}\begin{bmatrix}
\boldsymbol{I}^{N\times N} \\
-\boldsymbol{D}
\end{bmatrix} \hspace{1.0cm} \left( 0\leqslant t \leqslant T\right).
$$
\end{lemma}
\begin{proof}
We will apply \cite[Theorem 2.3]{freiling2000non}. We state this result here for convenience. It assures that, if we can find two matrices $\boldsymbol{Z}_1,\boldsymbol{Z}_2 \in \mathbb{R}^{N\times N},$ with $\boldsymbol{Z}_1$ symmetric, such that 
$$
\boldsymbol{Z}_1 - \boldsymbol{Z}_2 \boldsymbol{D} - \boldsymbol{D}^\intercal \boldsymbol{Z}_2 > \boldsymbol{0}^{N \times N},
$$
and for which the matrix
$$
\boldsymbol{L} := \begin{bmatrix}
\boldsymbol{Z}_1\boldsymbol{B}_{11} + \boldsymbol{Z}_2\boldsymbol{B}_{21} & \boldsymbol{Z}_1\boldsymbol{B}_{12} + \boldsymbol{B}_{11}^\intercal \boldsymbol{Z}_2 + \boldsymbol{Z}_2\boldsymbol{B}_{22} \\
\boldsymbol{0}^{N\times N} & \boldsymbol{B}_{12}\boldsymbol{Z}_2
\end{bmatrix},
$$
satisfies
$$
\boldsymbol{L} + \boldsymbol{L}^\intercal \leqslant \boldsymbol{0}^{2N \times 2N},
$$
then the Riccati ODE \eqref{eq:RiccatiODE} has a continuous solution $\boldsymbol{G} : \left[0,T\right] \rightarrow \mathbb{R}^{N\times N},$ which is differentiable for $t<T,$ and the formula \eqref{eq:FormulaIacopoRiccati} holds.

We take $\boldsymbol{Z}_1 = \boldsymbol{0}^{N\times N}$ and $\boldsymbol{Z}_2 = -\boldsymbol{I}^{N\times N}.$ It is clear that
$$
\boldsymbol{Z}_1 - \boldsymbol{Z}_2 \boldsymbol{D} - \boldsymbol{D}^\intercal \boldsymbol{Z}_2 = 2\boldsymbol{D} > \boldsymbol{0}^{N \times N}.
$$
Moreover, since
$$
\boldsymbol{L} = \begin{bmatrix}
-\boldsymbol{\Sigma} & \boldsymbol{C} \\
\boldsymbol{0}^{N\times N} & - \boldsymbol{K}^{-1}
\end{bmatrix},
$$
we have
$$
\boldsymbol{L} + \boldsymbol{L}^\intercal \leqslant - \left( c_1\wedge c_2 \right) \boldsymbol{I}^{2N \times 2N} \leqslant 0,
$$
where $c_1$ and $c_2$ are as in Remark \ref{rem:KeyRemark}.  
\end{proof}

Let us consider $\boldsymbol{\Pi}$ and $\boldsymbol{\Psi}$ as the solutions of the ODEs
$$
\begin{cases}
\dot{\boldsymbol{\Pi}} = - \boldsymbol{\Pi} \boldsymbol{K}^{-1}\boldsymbol{G},\\
\boldsymbol{\Pi}(0) = \boldsymbol{I}^{N\times N}
\end{cases}
$$
and
$$
\begin{cases}
\dot{\boldsymbol{\Psi}} = \boldsymbol{\Psi}\left( \boldsymbol{K}^{-1}\boldsymbol{G} + \boldsymbol{C} \right),\\
\boldsymbol{\Psi}(0) = \boldsymbol{I}^{N\times N}.
\end{cases}
$$

\begin{theorem}
Let us suppose that the assumptions of Theorem \ref{thm:NashEquilibriumGeneralSetting} are valid. Then, the Nash equilibrium $\boldsymbol{\nu}^*$ and the corresponding inventory $\boldsymbol{q}^*$ are both deterministic and admit the representation
$$
\begin{cases}
\boldsymbol{\nu}^*(t) = \boldsymbol{K}^{-1}\boldsymbol{g}_0(t) + \boldsymbol{K}^{-1}\boldsymbol{G}(t)\boldsymbol{q}^*(t),\\
\boldsymbol{q}^*(t) = \boldsymbol{\Pi}(t)^{-1}\left(\boldsymbol{q}_0 + \int_0^t \boldsymbol{\Pi}(u)\boldsymbol{K}^{-1}\boldsymbol{g}_0(u)\,du \right),
\end{cases}
$$
where
$$
\boldsymbol{g}_0(t) = \boldsymbol{\Psi}(t)^{-1}\int_t^T \boldsymbol{\Psi}(u)\boldsymbol{\mu}\,du.
$$
\end{theorem}
\begin{proof}
We write $\boldsymbol{x}^* := \boldsymbol{K}\boldsymbol{\nu}^*.$ It is straightforward to derive that $\boldsymbol{q}^*$ defined as in the statement satisfies
$$
\dot{\boldsymbol{q}}^* = \boldsymbol{K}^{-1}\boldsymbol{x}^*.
$$

Moreover,
\begin{align*}
    \dot{\boldsymbol{x}}^* &= \dot{\boldsymbol{g}}_0 + \dot{\boldsymbol{G}}\boldsymbol{q}^* + \boldsymbol{G}\dot{\boldsymbol{q}}^* \\
    &=\dot{\boldsymbol{g}}_0 + \boldsymbol{G}\boldsymbol{K}^{-1}\boldsymbol{g}_0 + \left( \dot{\boldsymbol{G}} + \boldsymbol{G}\boldsymbol{K}^{-1}\boldsymbol{G}\right)\boldsymbol{q}^* \\
    &= -\boldsymbol{C}\boldsymbol{g}_0 -\boldsymbol{\mu} + \left( -\boldsymbol{C}\boldsymbol{G} + \boldsymbol{\Sigma}\right)\boldsymbol{q}^* \\
    &= -\boldsymbol{C}\boldsymbol{x}^* + \boldsymbol{\Sigma} \boldsymbol{q}^* -\boldsymbol{\mu}.
\end{align*}

We conclude that $\left( \boldsymbol{q}^*,\boldsymbol{x}^*,\boldsymbol{0}\right) \in \mathbb{S}_{(N)}\times \mathbb{S}_{(N)}\times \mathbb{M}_{(N)}$ does indeed solve \eqref{eq:FBSDEHeterog}.  
\end{proof}
We provide in Figures \ref{fig:HeterogeneousIllustrations} and \ref{fig:Comparison_Heterog} some illustrations of the optimal trading game dynamics. In Tables \ref{tab:params2PlayerGame} and \ref{tab:params5PlayerGame}, we describe the parameters we used. In the former, the first player has an initial inventory equal to one, and the second one begins with no holdings. Alternatively, we can interpret that the agent $2$ as an arbitrageur, i.e., she has no initial target to execute, and is present in the market only to seize arbitrage opportunities. We see that trader $2$ takes advantage of the pressure exerted on the price by agent $1.$ An analogous situation occurs with the player $5$ in the game with five traders. In Figure \ref{fig:Comparison_Heterog}, we showcase the difference in the behavior of the agent $1.$ Although this trader has the same parameters and initial data in the two settings, her strategy is not the same due to the interactions with the other ones.

\begin{table}[!htp]
\centering
\begin{tabular}{@{}cccccc@{}}
\toprule
$i$ & $\alpha^i$         & $\kappa^i$           & $\lambda^i$        & $A^i$             & $q^i_0$  \\ \midrule
$1$ & $5 \times 10^{-5}$ & $2.5 \times 10^{-5}$ & $5\times 10^{-6}$  & $5\times 10^{-1}$ & $1$ \\
$2$ & $5 \times 10^{-6}$ & $2.5 \times 10^{-6}$ & $5 \times 10^{-7}$ & $2\times 10^{-2}$ & $0$ \\ \bottomrule
\end{tabular}
\caption{Parameters used in the $2-$player game. We fix $T=10$ and $\boldsymbol{\mu} = 0.$}
\label{tab:params2PlayerGame}
\end{table}

\begin{table}[!htp]
\centering
\begin{tabular}{@{}cccccc@{}}
\toprule
$i$ & $\alpha^i$         & $\kappa^i$           & $\lambda^i$        & $A^i$              & $q^i_0$ \\ \midrule
$1$ & $5 \times 10^{-5}$ & $2.5 \times 10^{-5}$ & $5\times 10^{-6}$  & $5\times 10^{-1}$  & $1$ \\
$2$ & $3 \times 10^{-5}$ & $1.5 \times 10^{-5}$ & $7 \times 10^{-7}$ & $3 \times 10^{-1}$ & $7 \times 10^{-1}$\\
$3$ & $2 \times 10^{-5}$ & $10^{-5}$            & $5 \times 10^{-7}$ & $2 \times 10^{-1}$ & $5 \times 10^{-1}$\\
$4$ & $5 \times 10^{-6}$ & $2.5 \times 10^{-6}$ & $2 \times 10^{-8}$ & $10^{-2}$          & $0$\\
$5$ & $10^{-5}$          & $5 \times 10^{-6}$   & $5 \times 10^{-7}$ & $10^{-1}$          & $-2 \times 10^{-1}$\\ \bottomrule
\end{tabular}
\caption{Parameters used in the $5-$player game. We fix $T=10.$}
\label{tab:params5PlayerGame}
\end{table}

\begin{figure}[!htp]
    \centering
    \includegraphics[scale = 0.35]{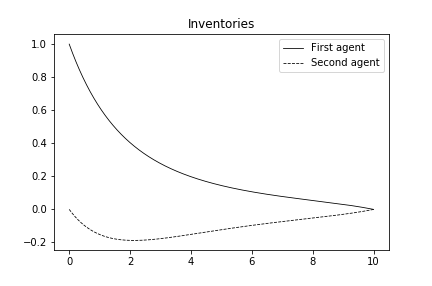}
    \includegraphics[scale = 0.35]{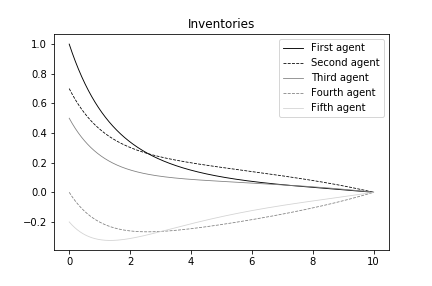}
    \caption{In the left panel, we present the inventories in the $2-$player game. In the right one, the corresponding objects in a population of $5$ agents.}
    \label{fig:HeterogeneousIllustrations}
\end{figure}

\begin{figure}[!htp]
    \centering
    \includegraphics[scale = 0.4]{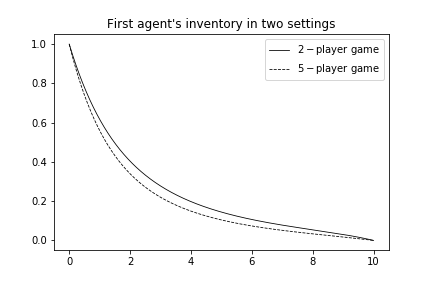}
    \caption{Comparison of the inventory of the first trader in the $2-$ and $5-$player game.}
    \label{fig:Comparison_Heterog}
\end{figure}

Now, let us suppose further that the parameters are homogeneous throughout the population. Under this assumption, we drop the superscripts: $\alpha^i \equiv \alpha, \kappa^i \equiv \kappa, \lambda^i\equiv \lambda$ and $A^i\equiv A.$ We also assume $\mu \equiv 0.$

We can find the analysis of the MFG model in this setting in \cite{cardaliaguet2017mean}. From their results, we know that the mean-field inventory $E$ solves the ODE system
\begin{equation} \label{eq:ODEsystemCL}
    \begin{cases}   
        2\kappa \ddot{E} + \alpha \dot{E} - 2\lambda E = 0,\\
        E(0) = E_0,\, \kappa \dot{E}(T) + A E(T) = 0,
    \end{cases}
\end{equation}
for some given initial data $E_0.$  This system has a closed-form solution, see \cite[Proposition 3.1]{cardaliaguet2017mean}, which we write below in an alternative form:
\begin{proposition} \label{prop:BenchmarkCL}
The solution $E$ of \eqref{eq:ODEsystemCL} is given by
$$
E(t) = E_0 \frac{y(t)}{y(0)},
$$
where
$$
y(t) := -\left( r^{-} + \frac{A}{\kappa} \right) \frac{e^{-r^{+}(T-t)}}{2\theta} + \left( r^{+} + \frac{A}{\kappa} \right)\frac{e^{-r^{-}(T-t)}}{2\theta},
$$
for the parameters
$$
\theta = \frac 1 \kappa \sqrt{\kappa \lambda +\frac{\alpha^2}{16}}  \text{ and }  r_{\pm}= -\frac{\alpha }{4 \kappa } \pm \theta.
$$
\end{proposition}

In what regards the $N-$player game, we can obtain a characterization of 
\begin{equation} \label{eq:AvgInventoryNplayer}
    E_N := \frac{1}{N}\sum_{i=1}^N q^{*i}
\end{equation}
as a solution of an ODE system very close to \eqref{eq:ODEsystemCL}. Furthermore, we can derive a closed-form formula for it. These two facts comprise the content of the next result.
\begin{theorem} \label{thm:ConstParamsHomog_Unconstrained}
When we assume that all model parameters are constant, as well as homogeneous throughout the population, and that model parameters satisfy the condition exposed in Lemma \ref{lem:concavity}, the average inventory holdings \eqref{eq:AvgInventoryNplayer} in the $N-$player game solves the ODE system
\begin{equation} \label{eq:ODEsystemNplayer}
    \begin{cases}   
        2\kappa \ddot{E}_N + \alpha\left(1-\frac{1}{N}\right) \dot{E}_N -2\lambda E_N = 0,\\
        E_N(0) = \frac{1}{N}\sum_{i=1}^N q^i_0,\, \kappa E_N^\prime(T) + A E_N(T) = \frac{\alpha}{2N}E_N(T).
    \end{cases}
\end{equation}
Furthermore, it is explicitly given by
$$
E_N(t) = \left(\frac{1}{N}\sum_{i=1}^{N} q^i_0 \right)\frac{y_N(t)}{y_N(0)},
$$
where 
\begin{align} \label{eq:yNforRepFinHomog}
  \begin{split}
    y_N(t) := -\left[ r^{-}_N + \frac{1}{\kappa}\left(A - \frac{\alpha}{2N} \right) \right] \frac{e^{-r^{+}_N(T-t)}}{2\theta_N} + \left[ r^{+}_N + \frac{1}{\kappa}\left(A - \frac{\alpha}{2N} \right) \right]\frac{e^{-r^{-}_N(T-t)}}{2\theta_N},
  \end{split}
\end{align}
for the parameters
$$
\theta_N := \frac{1}{4\kappa} \sqrt{ \alpha^2\left(1 - \frac{1}{N} \right)^2 + 16\lambda \kappa}
$$
and
$$
r^{\pm}_N := -\frac{\alpha}{4\kappa}\left(1 - \frac{1}{N}\right) \pm \theta_N.
$$
\end{theorem}
\begin{proof}
Let us write
$$
F_N := \frac{1}{N}\sum_{i=1}^N\nu^{*i}.
$$
We recall that, under the current assumptions, both $E_N$ and $F_N$ are deterministic. Moreover, system \eqref{eq:FBSDEHeterog} implies
\begin{equation} \label{eq:ODEsystemNplayerProof}
    \begin{cases}
    \dot{E}_N = \frac{1}{2\kappa}F_N,\\
    \dot{F}_N = \frac{\alpha}{2\kappa}\left(1-\frac{1}{N}\right)F_N + \lambda E_N,\\
    E_N(0) = \frac{1}{N}\sum_{i=1}^N q^i_0,\,
    -F_N(T) = \left( 2A - \frac{\alpha}{N} \right) E_N(T).
    \end{cases}
\end{equation}
We can easily see that system \eqref{eq:ODEsystemNplayerProof} is equivalent to \eqref{eq:ODEsystemNplayer}. The proof of the explicit formula \eqref{eq:yNforRepFinHomog} can be done as in \cite[Proposition 3.1]{cardaliaguet2017mean}.  
\end{proof}

In Figure \ref{fig:HomogeneousIllustration}, with the parameters of Table \ref{tab:Homogeneous_paramsNdata}, we illustrate the convergence of the $N-$player average inventory holdings, under the NE dynamics, to the corresponding MFG studied in \cite{cardaliaguet2017mean}.

\begin{table}[!htp]
\centering
\begin{tabular}{@{}ccccc@{}}
\toprule
$\alpha$  & $\kappa$  & $\lambda$ & $A$       & $E_0$ \\ \midrule
$10^{-5}$ & $10^{-5}$ & $10^{-7}$ & $10^{-1}$ & $1$   \\ \bottomrule
\end{tabular}
\caption{Parameters we used to compute $E_N$ and $E.$ We took $T=10$ and $E_N(0) = E(0) = E_0.$}
\label{tab:Homogeneous_paramsNdata}
\end{table}

\begin{figure}[!htp]
    \centering
    \includegraphics[scale = 0.4]{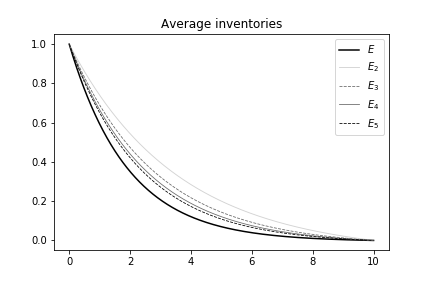}
    \caption{Average inventories $E_N,\, N\in\left\{1,2,3,4,5\right\},$ and the limiting mean-field one $E.$}
    \label{fig:HomogeneousIllustration}
\end{figure}

\section{Analysis of the \texorpdfstring{$N-$}{N-}player game: the constrained setting} \label{sec:AnalysisNplayerConstrained}

\subsection{Solving the constrained problem as an asymptotic limit}

In what follows, we state an auxiliary result which we will use many times.

\begin{lemma} \label{lem:TestingAgainstA0}
A stochastic process $\boldsymbol{M} = \left( M^1,...,M^N\right)^\intercal \in \mathbb{A}_{(N)}$ satisfies
\begin{equation}\label{eq:Martingality}
    \mathbb{E}\left[\int_0^T \boldsymbol{M}_t\cdot \boldsymbol{w}_t\,dt \right] = 0, \text{ for every } \boldsymbol{w} \in \mathbb{A}_{c,(N)}^0,
\end{equation}
if, and only if, $ M^i \in \mathcal{M}_i,$ for each $i\in \mathcal{N}.$
\end{lemma}
\begin{proof}
It is immediate that \eqref{eq:Martingality} is equivalent to the identities
\begin{equation} \label{eq:Martingality2}
    \mathbb{E}\left[\int_0^T M^i_t w^i_t\,dt \right] = 0 \hspace{1.0cm} \left(i\in\mathcal{N},\, w^i \in \mathbb{A}_{c,i}^0\right)
\end{equation}
holding simultaneously. We can now proceed as in \cite[Lemma 5.3]{bank2017hedging}, employing the Lebesgue Differentiation Theorem \cite[Theorem 3.21]{folland1999real} instead of using right-continuity, to conclude that \eqref{eq:Martingality2} is in turn equivalent to $ M^i \in \mathcal{M}_i,$ for each $i\in\mathcal{N}.$  
\end{proof}

In our next result, we obtain a characterization of Nash equilibria for the constrained problem. Together with the corresponding inventories, the NE solves an FBSDE system close to \eqref{eq:FBSDEHeterog}. 
\begin{proposition} \label{prop:CharacterizationNashConstrained}
Under the assumptions of Lemma \ref{lem:concavity}, a strategy $\boldsymbol{\nu}^* \in \mathbb{A}_{c,(N)}$ is a Nash equilibrium of the constrained problem if, and only if, there are $\boldsymbol{q}^* \in \mathbb{S}_{(N)}$ and $\boldsymbol{M}^* \in \mathcal{M}_{(N)}$ such that the processes $(\boldsymbol{q}^*, \boldsymbol{\nu}^*,\boldsymbol{M}^*)$ solve the FBSDE
    \begin{equation} \label{eq:FBSDEHeterogConstrained}
        \begin{cases}
        \boldsymbol{x}^*_t = \boldsymbol{K}_t\boldsymbol{\nu}_t^*,\\
        d\boldsymbol{q}^*_t = \boldsymbol{K}_t^{-1}\boldsymbol{x}^*_t\,dt,\\
        -d\boldsymbol{x}^*_t = \boldsymbol{\mathcal{P}}_t\left( \boldsymbol{C}_t\boldsymbol{x}^*_t \right)\,dt - \boldsymbol{\Sigma}_t \boldsymbol{q}^*_t\,dt + \frac{1}{N}\boldsymbol{\beta}_t\boldsymbol{q}^*_t\,dt + \boldsymbol{\mu}_t\,dt - d\boldsymbol{M}^*_t,\\
        \boldsymbol{q}^*_0 = \boldsymbol{q}_0 \text{ and } \boldsymbol{q}^*_T = \boldsymbol{0}.
        \end{cases}
    \end{equation}
\end{proposition}
\begin{proof}
When we restrict the admissible strategies of player $i$ to those in $\mathbb{A}_{c,i},$ we still have $\nu^i + \epsilon w^i \in \mathbb{A}_{c,i},$ for $\nu^i \in \mathbb{A}_{c,i}$ and $\epsilon \in \mathbb{R},$ as long as $w^i \in \mathbb{A}_{c,i}^0.$ For each of these $\nu^i$ and $w^i,$ the corresponding G\^ateaux derivative $\left\langle D_iJ_i\left(\nu^i;\boldsymbol{\nu}^{-i}\right), w^i\right\rangle$ exists. Proceeding as in the proof of Lemma \ref{lem:GateauxDerivWellDefnd}, we deduce that it is given by
\begin{align}
    \begin{split}
        \left\langle D_i J_i\left( \nu^i;\boldsymbol{\nu}^{-i} \right), w^i \right\rangle = \mathbb{E}\Bigg[\int_0^T w^i_t\Bigg\{&-2\kappa^i_t \nu^i_t -2 \int_t^T \lambda^i_u q^i_u\,du + \frac{\alpha^i_t}{N}q^i_t\,du \\
        & + \int_t^T \frac{\alpha^i_u}{N}\sum_{j= 1}^N  \nu^j_u\,du + \int_t^T \mu_u\,du \Bigg\}\,dt \Bigg].
    \end{split}
\end{align}
Since $J_i$ restricted to $\Pi_{j=1}^N \mathbb{A}_{c,j}$ remains strictly concave in the $i-$th direction, we can deduce, in the same way as we proved Corollaries \ref{cor:PartialOptmFBSDE} and \ref{cor:CharactNashHeterogReg}, that $\boldsymbol{\nu}^* \in \mathbb{A}_{c,i}$ is a Nash equilibrium for the constrained problem if, and only if, 
\begin{equation} \label{eq:CondnGateuxDerivConstrained}
    \left\langle D_i J_i\left( \nu^{*i};\boldsymbol{\nu}^{*-i} \right), w^i \right\rangle = 0 \hspace{1.0cm} \left(w^i \in \mathbb{A}_{c,i}^0,\, i \in \mathcal{N} \right).
\end{equation}
With the aid of Lemma \ref{lem:TestingAgainstA0}, we conclude that \eqref{eq:CondnGateuxDerivConstrained} is in turn equivalent to $\left(\boldsymbol{q}^*,\boldsymbol{x}^*:= \boldsymbol{K}\boldsymbol{\nu}^*,\boldsymbol{M}^*\right)$ solving \eqref{eq:FBSDEHeterogConstrained}, for some $ \boldsymbol{M}^* \in \mathcal{M}_{(N)}.$  
\end{proof}

We aim to obtain a solution of the constrained problem as an asymptotic limit of the processes solving the unconstrained one, relative to terminal penalty parameters tending to infinity. This is the content of the next Theorem.


\begin{theorem} \label{thm:NEStrictLiquidation}
    Let us suppose that
    $$
    \overline{\alpha}^2 < 4\underline{\lambda}\, \underline{\kappa} \text{ and } \underline{D} > 0.
    $$
    For each $\boldsymbol{A} = (A^1,...,A^N)^\intercal,$ we denote by $\boldsymbol{\nu}^{\boldsymbol{A}}$ the Nash equilibrium corresponding to these terminal penalty coefficients, and by $\boldsymbol{q}^{\boldsymbol{A}}$ the corresponding inventories. Then, there exist a sequence $\boldsymbol{A}^k \rightarrow \left(\infty,\ldots,\infty\right)^\intercal,$ as $k \rightarrow \infty,$ and processes $\boldsymbol{\nu}^\infty \in \mathbb{A}_{c,(N)}$ and $\boldsymbol{q}^\infty \in \mathbb{S}_{(N)}$ for which
    $$
    \boldsymbol{\nu}^{\boldsymbol{A}^k} \rightharpoonup \boldsymbol{\nu}^\infty \text{ and } \boldsymbol{q}^{\boldsymbol{A}^k} \rightharpoonup \boldsymbol{q}^\infty,
    $$
    weakly in the topology of $\mathbb{A}_{(N)},$ as $k \rightarrow \infty.$ Furthermore, $\left( \boldsymbol{\nu}^\infty,\,\boldsymbol{q}^\infty,\,\boldsymbol{M}^\infty\right) $ solves \eqref{eq:FBSDEHeterogConstrained}, for a suitable $\boldsymbol{M}^\infty \in \mathcal{M}_{(N)};$ hence, $\boldsymbol{\nu}^\infty$ is a Nash equilibrium for the constrained problem. 
\end{theorem}
\begin{proof}
Let us fix $\nu^{\text{TWAP},i} \in \mathbb{A}_{c,i} \subseteq \mathbb{A}_{i},$ $i\in\mathcal{N},$ where
$$
\nu^{\text{TWAP},i}_t := -\frac{q^i_0}{T} \hspace{1.0cm} \left(0\leqslant t \leqslant T\right).
$$
Optimality of $\boldsymbol{\nu}^{\boldsymbol{A}}$ gives
\begin{equation} \label{eq:PassingLimit1}
    J_i\left( \nu^{\boldsymbol{A},i};\boldsymbol{\nu}^{\boldsymbol{A},-i}\right) \geqslant J_i\left( \nu^{\text{TWAP},i}; \boldsymbol{\nu}^{\boldsymbol{A},-i} \right).
\end{equation}

On the one hand, it is straightforward to estimate
\begin{align} \label{eq:PassingLimit2}
    \begin{split}
        J_i( \nu^{\text{TWAP},i}; \boldsymbol{\nu}^{\boldsymbol{A},-i} ) \geqslant -C_i\left(\epsilon\right) - \frac{\epsilon}{N}\left( 1 - \frac{1}{N} \right)\mathbb{E}\left[\int_0^T \sum_{j\neq i} \left(\nu^{\boldsymbol{A},j}_t\right)^2\,dt \right],
    \end{split}
\end{align}
for each $i \in \mathcal{N}$ and each $\epsilon > 0,$ where
\begin{align*}
    \begin{split}
        C_i\left(\epsilon\right):= \mathbb{E}\Bigg[\int_0^T \Bigg\{ \left(\frac{q^i_0}{T} \right)^2\left[\kappa^i_t + \left( \lambda^i_t + \frac{\left(\alpha^i_t\right)^2}{4\epsilon}\right)\left(T-t \right)^2 + \frac{\alpha^i_t}{N}(T-t) \right] - \frac{q^i_0}{T}\left(T-t\right)\mu^i_t \Bigg\}\,dt \Bigg].
    \end{split}
\end{align*}

On the other hand, for $a >0,$ we have
\begin{align} \label{eq:PassingLimit3}
    \begin{split}
        J_i ( \nu^{\boldsymbol{A},i};\boldsymbol{\nu}^{\boldsymbol{A},-i}) \leqslant-\mathbb{E}\Bigg[\int_0^T&\left\{\kappa^i_t\left( \nu_t^{\boldsymbol{A},i}\right)^2 + \left( \lambda^i_t - \frac{\left(\alpha^i_t\right)^2}{2a^2} - \epsilon\right)\left(q_t^{\boldsymbol{A},i}\right)^2 \right\}\,dt  \\
        &\hspace{0.3cm}+ A^i\left(q^{\boldsymbol{A},i}_T\right)^2 \Bigg] + \frac{a^2}{2N}\mathbb{E}\left[\int_0^T\sum_{j=1}^N \left(\nu^{\boldsymbol{A},j}_t \right)^2\,dt \right]+ \frac{1}{4\epsilon}\|\mu^i\|^2.
    \end{split}
\end{align}
Synthesizing \eqref{eq:PassingLimit1}-\eqref{eq:PassingLimit3} upon summing them over $i\in \mathcal{N},$ we deduce
\begin{align} \label{eq:SecondToLast_LimitInA}
    \begin{split}
        \mathbb{E}\Bigg[\int_0^T&\sum_{i=1}^N \Bigg\{ \kappa^i_t - \frac{a^2}{2} - \epsilon\left(1-\frac{1}{N} \right)^2 \Bigg\}\left(\nu^{\boldsymbol{A}, i}_t \right)^2\,dt\Bigg] \\
        &+ \mathbb{E}\Bigg[\int_0^T\sum_{i=1}^N \left(\lambda^i_t - \frac{\left(\alpha^i_t \right)^2}{2a^2} -\epsilon \right)\left( q^{\boldsymbol{A},i}_t \right)^2\,dt + \sum_{i=1}^N A^i\left(q^{\boldsymbol{A},i}_T\right)^2\Bigg] \\
        &\hspace{3.65cm}\leqslant \sum_{i=1}^N C_i\left(\epsilon \right) + \frac{1}{2\epsilon}\sum_{i=1}^N\|\mu^i\|^2 =: R_N(\epsilon).
    \end{split}
\end{align}

We take $a$ such that
$$
\frac{ \overline{\alpha}^2}{2\underline{\lambda} }< a^2 <2\underline{\kappa} , 
$$
we fix $\epsilon > 0$ sufficiently small, and we assume $A^i \geqslant A > 0,$ for all $i\in \mathcal{N},$ concluding from inequality \eqref{eq:SecondToLast_LimitInA} what follows
\begin{equation} \label{eq:MainEstimate_LimitInA}
    \left\|\boldsymbol{\nu}^{\boldsymbol{A}}\right\|^2 + \left\|\boldsymbol{q}^{\boldsymbol{A}}\right\|^2 + A\left\|\boldsymbol{q}^{\boldsymbol{A}}_T\right\|_{L^2\left(\Omega,\mathcal{F}\right)^N}^2 \leqslant C R_N(\epsilon),
\end{equation}
where $C > 0$ can be taken to be independent of $\boldsymbol{A}$ (and $N$).

Let us pass to a subsequence $\boldsymbol{A}^k = \left( A^{1,k},...,A^{N,k} \right)^\intercal \rightarrow \left(\infty,...,\infty\right)^\intercal$ such that
\begin{equation} \label{eq:Convergences_LimitInA}
    \begin{cases}
        \boldsymbol{\nu}^{\boldsymbol{A}^k} \rightharpoonup \boldsymbol{\nu}^\infty, \text{ weakly in } \mathbb{A}_{(N)},\\
        \boldsymbol{q}^{\boldsymbol{A}^k} \rightharpoonup \boldsymbol{q}^\infty, \text{ weakly in } \mathbb{A}_{(N)},\\
        \boldsymbol{q}^{\boldsymbol{A}^k}_T \rightarrow 0, \text{ strongly in } L^2\left(\Omega,\mathcal{F}\right)^N.
    \end{cases}
\end{equation}
We claim $\left(\boldsymbol{q}^\infty, \boldsymbol{x}^\infty := \boldsymbol{K}\boldsymbol{\nu}^\infty, \boldsymbol{M}^\infty\right)$ solves \eqref{eq:FBSDEHeterogConstrained}, for some $\boldsymbol{M}^\infty \in \mathcal{M}_{(N)}.$ In effect, the relation
$$
\frac{d}{dt}\boldsymbol{q}^{\boldsymbol{A}^k} = \boldsymbol{\nu}^{\boldsymbol{A}^k}, \text{ distributionally, for } k\geqslant 1,
$$
implies
\begin{equation} \label{eq:ODEpartOfBSDE_LimitInA}
    \frac{d}{dt}\boldsymbol{q}^{\infty} = \boldsymbol{\nu}^{\infty} = \boldsymbol{K}^{-1}\boldsymbol{x}^\infty, \text{ distributionally.}
\end{equation}
Consequently, the paths $t \mapsto \boldsymbol{q}^\infty_t$ are absolutely continuous $\mathbb{P}-$a.s., whence \eqref{eq:ODEpartOfBSDE_LimitInA} holds $\mathbb{P}-$a.s. for almost every $t \in \left[0,T \right],$ as well as $\boldsymbol{q}_T^\infty = 0,$ $\mathbb{P}-$a.s. (by the last convergence in \eqref{eq:Convergences_LimitInA}). Thus, the membership $\boldsymbol{\nu}^\infty \in \mathbb{A}_{c,(N)}$ holds.

Let us take the dot product of the BSDE part of \eqref{eq:FBSDEHeterog} with a given $\boldsymbol{w} \in \mathbb{A}_{c,(N)}^0,$ and then integrate the result over $\left[0,T\right]\times \Omega$ against $dt\times d\mathbb{P},$ from where it follows that
\begin{align} \label{eq:TestedFBSDEinA_LimitInA}
    \begin{split}
        \mathbb{E} \Bigg[ \int_0^T \boldsymbol{x}^{\boldsymbol{A}^k}_t\cdot \boldsymbol{w}_t\,dt\Bigg] = \mathbb{E}\left[\int_0^T\int_t^T\left[ \boldsymbol{C}_u\boldsymbol{x}^{\boldsymbol{A}^k}_u - \left(\boldsymbol{\Sigma}_u - \frac{1}{N}\boldsymbol{\beta}_u\right)\boldsymbol{q}^{\boldsymbol{A}^k}_u + \boldsymbol{\mu}_u\right]\,du \cdot \boldsymbol{w}_t\,dt \right].
    \end{split}
\end{align}
Above, we employed Lemma \ref{lem:TestingAgainstA0} to ensure that
$$
\mathbb{E}\left[\int_0^T \boldsymbol{w}_t \cdot \boldsymbol{M}^{\boldsymbol{A}^k}_t\,dt\right] = 0,
$$
for all $k\geqslant 1,$ and we also used the fact that
$$
\mathbb{E}\left[\int_0^T \boldsymbol{x}^{\boldsymbol{A}^k}_T\cdot \boldsymbol{w}_t\,dt \right] = \mathbb{E}\left[\boldsymbol{x}^{\boldsymbol{A}^k}_T\cdot \int_0^T \boldsymbol{w}_t\,dt \right] = 0.
$$
From the convergences in \eqref{eq:Convergences_LimitInA}, we deduce
$$
\mathbb{E}\Bigg[ \int_0^T \boldsymbol{x}^{\boldsymbol{A}^k}_t\cdot \boldsymbol{w}_t\,dt\Bigg] \rightarrow \mathbb{E}\Bigg[ \int_0^T \boldsymbol{x}^{\infty}_t\cdot \boldsymbol{w}_t\,dt\Bigg],
$$
as $k\rightarrow \infty.$ Also, applying Fubini's Theorem and basic properties of the projections $\left\{ \mathcal{P}^i\right\}_{i\in\mathcal{N}},$
\begin{align*}
    \mathbb{E}\Bigg[\int_0^T &\int_t^T \boldsymbol{\mathcal{P}}_u\left(\boldsymbol{C}_u\boldsymbol{x}^{\boldsymbol{A}^k}_u\right) \,du \cdot \boldsymbol{w}_t\,dt \Bigg] \\
    &= \mathbb{E}\left[\int_0^T \boldsymbol{\nu}^{\boldsymbol{A}^k}_u \cdot \boldsymbol{\mathcal{P}}_u\left( \left(\boldsymbol{C}_u\boldsymbol{K}_u^{-1} \right)^\intercal \int_0^u \boldsymbol{w}_t\,dt \right)\,du \right] \\
    &\longrightarrow \mathbb{E}\left[\int_0^T \boldsymbol{\nu}^{\infty}_u \cdot \boldsymbol{\mathcal{P}}_u\left( \left(\boldsymbol{C}_u\boldsymbol{K}_u^{-1} \right)^\intercal \int_0^u \boldsymbol{w}_t\,dt \right)\,du \right] \\
    &= \mathbb{E}\left[\int_0^T \int_t^T \boldsymbol{\mathcal{P}}_u\left( \boldsymbol{C}_u\boldsymbol{x}^{\infty}_u \right) \,du \cdot \boldsymbol{w}_t\,dt \right],
\end{align*}
since $\left\{ \boldsymbol{\mathcal{P}}_u\left( \left(\boldsymbol{C}_u\boldsymbol{K}_u^{-1} \right)^\intercal \int_0^u \boldsymbol{w}_t\,dt \right) \right\}_{0\leqslant u \leqslant T} \in \mathbb{A}_{(N)},$ and likewise,
\begin{align*}
    \mathbb{E}\Bigg[\int_0^T \int_t^T \left(\boldsymbol{\Sigma}_u - \frac{1}{N}\boldsymbol{\beta}_u\right)\boldsymbol{q}^{\boldsymbol{A}^k}_u\,du \cdot \boldsymbol{w}_t\,dt \Bigg] \longrightarrow \mathbb{E}\left[\int_0^T \int_t^T \left(\boldsymbol{\Sigma}_u - \frac{1}{N}\boldsymbol{\beta}_u\right)\boldsymbol{q}^{\infty}_u\,du \cdot \boldsymbol{w}_t\,dt \right]
\end{align*}
whence, passing \eqref{eq:TestedFBSDEinA_LimitInA} to the limit as $k \rightarrow \infty,$ we obtain
\begin{align} \label{eq:DesiredBSDE_LimitInA}
    \begin{split}
        \mathbb{E}\Bigg[ \int_0^T \widetilde{\boldsymbol{M}}_t\cdot \boldsymbol{w}_t\,dt\Bigg] = 0,
    \end{split}
\end{align}
where
\begin{align*}
    \widetilde{\boldsymbol{M}}_t &:= \boldsymbol{x}^{\infty}_t - \boldsymbol{\mathcal{P}}_t\left( \int_t^T\left[ \boldsymbol{\mathcal{P}}_u\left( \boldsymbol{C}_u\boldsymbol{x}^{\infty}_u\right) - \left(\boldsymbol{\Sigma}_u - \frac{1}{N}\boldsymbol{\beta}_u\right)\boldsymbol{q}^{\infty}_u + \boldsymbol{\mu}_u \right]\,du \right) \\
    &= \boldsymbol{x}^{\infty}_t +  \int_0^t\left[ \boldsymbol{\mathcal{P}}_u\left( \boldsymbol{C}_u\boldsymbol{x}^{\infty}_u\right) - \left(\boldsymbol{\Sigma}_u - \frac{1}{N}\boldsymbol{\beta}_u\right)\boldsymbol{q}^{\infty}_u + \boldsymbol{\mu}_u \right]\,du - \boldsymbol{M}^\prime_t,
\end{align*}
where we have written
$$
\boldsymbol{M}^\prime_t := \boldsymbol{\mathcal{P}}_t\left( \int_0^T\left[ \boldsymbol{\mathcal{P}}_u\left( \boldsymbol{C}_u\boldsymbol{x}^{\infty}_u\right) - \left(\boldsymbol{\Sigma}_u - \frac{1}{N}\boldsymbol{\beta}_u\right)\boldsymbol{q}^{\infty}_u + \boldsymbol{\mu}_u \right]\,du \right).
$$
Applying Lemma \ref{lem:TestingAgainstA0}, we conclude that $\left(\boldsymbol{q}^\infty,\boldsymbol{x}^\infty,\boldsymbol{M}^\infty := \widetilde{\boldsymbol{M}} + \boldsymbol{M}^\prime \right)$ solves the FBSDE \eqref{eq:FBSDEHeterogConstrained}. 

\end{proof}

\begin{lemma}
    There exists at most one Nash equilibrium $\boldsymbol{\nu}$ for the constrained game satisfying $\boldsymbol{\nu} \in \mathbb{S}_{c,(N)}.$ 
\end{lemma}
\begin{proof}
    Let us show that the system \eqref{eq:FBSDEHeterogConstrained} admits at most one solution $\left(\boldsymbol{\nu},\boldsymbol{q},\boldsymbol{M}\right) \in \mathbb{S}_{c, (N)} \times \mathbb{S}_{(N)} \times \mathcal{M}_{(N)}.$ Once we establish this, we can promptly conclude that the unique solution to \eqref{eq:FBSDEHeterogConstrained} is the only Nash equilibrium for the constrained problem. In effect, let us assume $\left(\boldsymbol{q},\boldsymbol{x},\boldsymbol{M}\right)$ and $\left(\boldsymbol{q}^\prime,\boldsymbol{x}^\prime,\boldsymbol{M}^\prime\right)$ are two solutions of \eqref{eq:FBSDEHeterogConstrained} having this regularity. We set $\left(\Delta\boldsymbol{q},\Delta\boldsymbol{x},\Delta\boldsymbol{M}\right) := \left(\boldsymbol{q} - \boldsymbol{q}^\prime,\boldsymbol{x} - \boldsymbol{x}^\prime,\boldsymbol{M} - \boldsymbol{M}^\prime\right).$ Using It\^o's formula and conducting estimates similar to those carried out in the first part of the proof of Theorem \ref{thm:NashEquilibriumGeneralSetting}, we obtain
    \begin{align*}
        \mathbb{E}\left[\Delta \boldsymbol{q}_T \cdot \Delta\boldsymbol{x}_T - \Delta \boldsymbol{q}_0 \cdot \Delta\boldsymbol{x}_0 \right] &\geqslant c_1\|\Delta\boldsymbol{x}\|^2 + c_2\|\Delta\boldsymbol{q}\|^2.
    \end{align*}
    Since $\Delta \boldsymbol{q}_T = 0$ and, likewise, $\Delta \boldsymbol{q}_0 = 0,$ $\mathbb{P}-$a.s., we conclude that $\Delta \boldsymbol{q} \equiv 0$ and $\Delta \boldsymbol{x} \equiv 0,$ whence $\Delta \boldsymbol{M} \equiv 0$ as well. 
\end{proof}
\begin{remark}
For $\boldsymbol{\nu}^\infty$ as in Theorem \ref{thm:NEStrictLiquidation}, it is straightforward to check that the condition $\boldsymbol{\nu}^\infty \in \mathbb{S}_{c,(N)}$ is true if we assume constant coefficients under weak interaction, cf. \cite[Theorem 2.2]{schied2017state}.
\end{remark}


It is pertinent to register a further consequence of estimate \eqref{eq:MainEstimate_LimitInA}:

\begin{corollary} \label{cor:EstimateInN_Constrained}
Under the assumptions and notations of Theorem \ref{thm:NEStrictLiquidation}, we have
$$
\|\boldsymbol{\nu}^\infty\|^2 + \|\boldsymbol{q}^\infty\|^2 \leqslant CR_N(\epsilon),
$$
where $R_N$ is defined in \eqref{eq:SecondToLast_LimitInA}, $C$ and $\epsilon$ are fixed positive constants depending only on $\overline{\alpha},\, \underline{\kappa},\, \underline{\lambda}$ and $\left\{ \left\|\mu^i\right\| \right\}_{i=1}^N.$
\end{corollary}
\begin{proof}
From the uniform boundedness principle (see \cite[Proposition 3.5 (iii)]{brezis2010functional}), alongside estimate \eqref{eq:MainEstimate_LimitInA} mentioned above, we have
$$
\|\boldsymbol{\nu}^\infty\|^2 + \|\boldsymbol{q}^\infty\|^2 \leqslant \liminf_{k \rightarrow \infty } \left(\|\boldsymbol{\nu}^{\boldsymbol{A}^k}\|^2 + \|\boldsymbol{q}^{\boldsymbol{A}^k}\|^2 \right) \leqslant CR_N(\epsilon). 
$$  
\end{proof}

\subsection{Relating the constrained $N-$player game with the corresponding MFG}

Let us assume hypotheses of Theorem \ref{thm:NEStrictLiquidation} hold. We denote by $\left( \boldsymbol{q}^N, \boldsymbol{x}^N = \boldsymbol{K}\boldsymbol{\nu}^N, \boldsymbol{M}^N \right)$ a Nash equilibrium, as obtained in this Theorem, corresponding to a population size $N$ and $\boldsymbol{\mu} \equiv 0.$ 
We devote this subsection to investigating the convergence properties of the Nash equilibrium $\boldsymbol{\nu}^N$ as $N$ tends to infinity. Our first remark follows from Corollary \ref{cor:EstimateInN_Constrained}.
\begin{corollary} \label{cor:AvgBddnessInN}
If $\left\{q^i_0\right\}_{i\geqslant 1}$ is $L^2-$bounded, and the assumptions of Theorem \ref{thm:NEStrictLiquidation} are valid, then the Nash equilibrium $\boldsymbol{\nu}^N = \left(\nu^{N,1},...,\nu^{N,N} \right)^\intercal$ and corresponding inventory $\boldsymbol{q}^N = \left(q^{N,1},...,q^{N,N}\right)^\intercal$ satisfy
$$
\sup_N \left\{ \frac{1}{N} \sum_{i=1}^N \left( \left\|\nu^{N,i} \right\|^2 + \left\|q^{N,i} \right\|^2 \right) \right\} < \infty.
$$
\end{corollary}

For the remainder of the present section, we stipulate that the hypotheses of Theorem \ref{thm:NEStrictLiquidation} are in force. Moreover, we assume the following:
\begin{itemize}
    \item The probability space $\left(\Omega,\mathcal{F},\mathbb{P}\right)$ supports standard independent Brownian motions $\boldsymbol{W},\boldsymbol{W}^0,\boldsymbol{W}^1,\boldsymbol{W}^2,...,$ where $\boldsymbol{W}^0$ is $d^0-$dimensional, and all others are $d-$dimensional, for two given integers $d^0,d\geqslant 1;$
    \item Initial inventories $\left\{q^i_0\right\}_{i\geqslant 1}$ are independent and identically distributed;
    \item For each $i\in \mathcal{N},$ the filtration $\mathbb{F}^i$ is the $\mathbb{P}-$augmentation of the filtration generated by $q^i_0,$ $W^0$ and $W^i;$
    \item There exist measurable functions $\alpha,\kappa,\lambda$ and $\beta$ such that
    $$
    \alpha^i_t = \alpha\left( q^i_0, W^0_{\cdot\wedge t}, W^i_{\cdot\wedge t}\right),\, \kappa^i_t = \kappa\left( q^i_0, W^0_{\cdot\wedge t}, W^i_{\cdot\wedge t}\right),\, \lambda^i_t = \lambda\left( q^i_0, W^0_{\cdot\wedge t}, W^i_{\cdot\wedge t}\right)
    $$
    and 
    $$
    \beta^i_t = \beta\left( q^i_0, W^0_{\cdot\wedge t}, W^i_{\cdot\wedge t}\right);
    $$
    \item The functions $\alpha,\kappa,\lambda$ and $\beta$ satisfy 
    $$
    \underline{\alpha} \leqslant \alpha \leqslant \overline{\alpha},\, \underline{\kappa} \leqslant \kappa \leqslant \overline{\kappa},\, \underline{\lambda} \leqslant \lambda \leqslant \overline{\lambda} \text{ and } \left|\beta\right|\leqslant \overline{\beta};
    $$
    \item The constants $\underline{\kappa},\, \underline{\lambda}$ and $\overline{\alpha}$ satisfy
    $$
    \overline{\alpha}^2 < 4 \underline{\lambda}\, \underline{\kappa}.
    $$
\end{itemize}
Let us also write $\mathbb{F}^0$ for the $\mathbb{P}-$augmentation of the filtration generated by $\boldsymbol{W}^0,$ and likewise, we denote by $\overline{\mathbb{F}}$ the one obtained through this process from $\boldsymbol{W}^0,\boldsymbol{W}$ and a square-integrable random variable $q_0$ independent of $\left\{ q^i_0 \right\}_{i\geqslant 1},$ having the same distribution of these.  

In this context, we put ourselves in the framework of \cite{fu2018mean}, where the subsequent MFG is studied: 
$$
\begin{cases}
1. \text{ fix an } \mathbb{F}^0-\text{progressively measurable process } \mu \in \mathbb{L}^2;\\
2. \text{ solve } \xi^* = \argmax_{\xi} \mathbb{E}\left[\int_0^T\left\{\alpha_t \mu_t q_t - \lambda_t \left(q_t\right)^2 - \kappa_t \left(\xi_t\right)^2 \right\}\,dt\right] ,\\
\hspace{0.4cm}\text{where } \xi \in \mathbb{L}^2 \text{ is } \overline{\mathbb{F}}-\text{progressively measurable, } dq_t = \xi_t\,dt,\, q_T = 0;\\
3. \text{ search for a fixed point } \mu^*_t = \mathbb{E}\left[\xi^*_t|\mathcal{F}^0_t\right], \text{ a.e.a.s.}
\end{cases}
$$
In this work, the authors prove that the above problem is equivalent to solving the FBSDE system 
\begin{equation} \label{eq:FBSDE_MFG}
    \begin{cases}
        dq^*_t = \frac{x^*_t}{2\kappa_t}\,dt,\\
        -dx^*_t = \alpha_t  \mathbb{E}\left[ \frac{x^*_t}{2\kappa_t} \Big|\mathcal{F}^0_t \right]\,dt - 2\lambda_t q^*_t\,dt - \boldsymbol{Z}^*_t\cdot d\widetilde{\boldsymbol{W}}_t,\\
        q^*_0 = q_0 \text{ and } q^*_T = 0,
    \end{cases}
\end{equation}
with corresponding optimal aggregation effect 
\begin{equation} \label{eq:MFG_rate}
    \mu^*_t = \mathbb{E}\left[ \frac{x^*_t}{2\kappa_t} \Big|\mathcal{F}^0_t \right],
\end{equation}
where we have written $\widetilde{W} = \left( \boldsymbol{W}^0, \boldsymbol{W}\right).$ We introduce the best response $\xi^{*,i}$ of player $i$ to the mean-field $\mu^*$ described in \eqref{eq:MFG_rate} as in \cite[Section 3]{fu2018mean}, i.e.,
$$
\xi^{*,i} = \argmax_{\xi^i \in \mathbb{A}^i} \mathbb{E}\left[\int_0^T\left\{\alpha^i_t \mu^*_t q^i_t - \lambda^i_t \left(q^i_t\right)^2 - \kappa^i_t \left(\xi^i_t\right)^2 \right\}\,dt\right] ,
$$
where $q^i$ is initiated at $q^i_0,$ and is constrained to $dq^i_t = \xi^i_t\,dt$ and $q^i_T = 0.$ We proceed to list the most important properties of the processes $\left\{ \xi^{*,i} \right\}_{i\geqslant 1}$ for our present purposes. 
\begin{lemma} \label{lem:BestResponsesProperties}
(a) For some $\boldsymbol{Z}^{*,i} \in \mathbb{A}_i,$ the process $\left(q^{*,i},x^{*,i} = 2\kappa^i\xi^{*,i},\boldsymbol{Z}^{*,i}\right)$ solves the FBSDE
$$
\begin{cases}
dq^{*,i}_t = \frac{x^{*,i}_t}{2\kappa^i_t}\,dt,\\
-dx^{*,i}_t = \alpha^i_t \mu^*_t\,dt - 2\lambda^i_t q^{*,i}_t\,dt - \boldsymbol{Z}^{*,i}_t\cdot d\widetilde{\boldsymbol{W}}^i_t,\\
q^{*,i}_0 = q_0 \text{ and } q^{*,i}_T= 0,
\end{cases}
$$
where $\widetilde{\boldsymbol{W}}^i = \left( \boldsymbol{W}^0, \boldsymbol{W}^i \right).$

(b) The relation 
$$
\mu^*_t = \mathbb{E}\left[\xi^{*,i}_t \Big|\mathcal{F}^0_t \right], \text{ a.e.a.s.,}
$$
holds.

(c) The estimate
$$
\left\| \frac{1}{N}\sum_{i=1}^N \xi^{*,i} - \mu^* \right\| = O\left(\frac{1}{\sqrt{N}}\right)
$$
is valid.
\end{lemma}
\begin{proof}
Item ($a$) is a particular case of the characterization of constrained Nash equilibria made in Proposition \ref{prop:CharacterizationNashConstrained} (we remark that the lack of interaction in the objective functionals decouples the system). The result of item ($b$) is established in \cite[Section 3]{fu2018mean}, and we can prove the estimate of ($c$) with the same iterated conditioning and conditional independence arguments used in the proof of \cite[Theorem 3.3]{fu2018mean} (c.f. equation (3.5) therein).  
\end{proof}

With Lemma \ref{lem:BestResponsesProperties} at hand, we are ready to prove the main result of this subsection.
\begin{theorem} \label{thm:ConvergenceToTheMFG}
    As long as $\boldsymbol{\nu}^N = \left( \nu^{N,1},...,\nu^{N,N}\right)^\intercal \in \mathbb{S}_{c,(N)},$ it satisfies
    $$
    \left\| \frac{1}{N}\sum_{i=1}^N \nu^{*,i} - \mu^* \right\| = O\left(\frac{1}{\sqrt{N}}\right).
    $$
\end{theorem}
\begin{proof}
We notice that, under the current assumptions, the system \eqref{eq:FBSDEHeterogConstrained} reads componentwise as
\begin{equation} \label{eq:ConvToTheMFG_eq1}
    \begin{cases}
    x^{N,i}_t = 2\kappa^i_t \nu^{N,i}_t,\\
    dq^{N,i}_t = \frac{x^{N,i}_t}{2\kappa^i_t}\,dt,\\
    -dx^{N,i}_t = \alpha^i_t \mathbb{E}\left[ \frac{1}{N}\sum_{j\neq i}\frac{x^{N,j}_t}{2\kappa^j_t} \Big| \mathcal{F}^0_t \right]\,dt - \left(2\lambda^i_t - \frac{1}{N}\beta^i_t\right)q^{N,i}_t\,dt - \boldsymbol{Z}^{N,i}\cdot d\widetilde{\boldsymbol{W}}^i_t,\\
    q^{N,i}_0 = q^i_0 \text{ and } q^{N,i}_T= 0.
    \end{cases}
\end{equation}
Using Lemma \ref{lem:BestResponsesProperties} ($b$), we see that it is licit to write
$$
\mu^*_t = \mathbb{E}\left[ \frac{1}{N}\sum_{j=1}^N \xi^{*,j}_t \Bigg| \mathcal{F}^0_t \right],
$$
in such a way that the optimality system that we have presented in Lemma \ref{lem:BestResponsesProperties} ($a$) becomes
\begin{equation} \label{eq:ConvToTheMFG_eq2}
    \begin{cases}
    x^{*,i}_t = 2\kappa^i_t \xi^{*,i}_t,\\
    dq^{*,i}_t = \frac{x^{*,i}_t}{2\kappa^i_t}\,dt,\\
    -dx^{*,i}_t = \alpha^i_t \mathbb{E}\left[ \frac{1}{N}\sum_{j=1}^N \frac{x^{*,j}_t}{2\kappa^j_t} \Big| \mathcal{F}^0_t \right]\,dt - 2\lambda^i_t q^{*,i}_t\,dt - \boldsymbol{Z}^{*,i}_t\cdot d\widetilde{\boldsymbol{W}}^i_t,\\
    q^{*,i}_0 = q_0 \text{ and } q^{*,i}_T= 0.
    \end{cases}
\end{equation}
Let us write 
$$
\Delta \boldsymbol{q}^N = \left( q^{N,1} - q^{*,1},...,q^{N,N} - q^{*,N}\right)^\intercal
$$
and 
$$
\Delta \boldsymbol{x}^N = \left( x^{N,1} - x^{*,1},...,x^{N,N} - x^{*,N} \right)^\intercal.
$$
Using It\^o's formula, together with \eqref{eq:ConvToTheMFG_eq1} and \eqref{eq:ConvToTheMFG_eq2}, and conducting estimates as in Theorem \ref{thm:NashEquilibriumGeneralSetting}, we obtain
\begin{align} \label{eq:MainEstimateForConvInN}
    \begin{split}
        0 &= \mathbb{E}\left[\Delta \boldsymbol{q}^N_T \cdot \Delta \boldsymbol{x}^N_T - \Delta \boldsymbol{q}^N_0 \cdot \Delta\boldsymbol{x}_0^N \right] \\
        &\geqslant c_1\|\Delta \boldsymbol{x}^N\|^2 + c_2 \|\Delta \boldsymbol{q}^N\|^2 - \frac{C}{N^2}\sum_{i=1}^N\left( \left\|x^{N,i}\right\|^2 + \left\|q^{N,i} \right\|^2\right)\\
        &\geqslant c_1\|\Delta \boldsymbol{x}^N\|^2 + c_2 \|\Delta \boldsymbol{q}^N\|^2 - \frac{C}{N^2}\sum_{i=1}^N\left( \left\|\nu^{N,i}\right\|^2 + \left\|q^{N,i} \right\|^2\right),
    \end{split}
\end{align}
where the positive constants $c_1,c_2$ and $C$ are independent of $N.$ Taking into account Corollary \ref{cor:AvgBddnessInN}, we infer from estimate \eqref{eq:MainEstimateForConvInN} that
\begin{equation} \label{eq:ConcludingEstForConvInN}
    \left\|\frac{1}{N}\sum_{i=1}^N \nu^{N,i} - \frac{1}{N}\sum_{i=1}^N \xi^{*,i}\right\|^2 \leqslant \frac{C}{N}\left\|\Delta \boldsymbol{x}^N \right\|^2 \leqslant \frac{C}{N^2},
\end{equation}
the constant $C>0$ being independent of $N.$ We conclude the proof of the Theorem by combining \eqref{eq:ConcludingEstForConvInN} with Lemma \ref{lem:BestResponsesProperties} ($c$).  
\end{proof}

\section{An extension: the hierarchical game} \label{sec:remarks_major_minor}

\subsection{Adding a major player to the game}

Here, we assume that, in addition to the population of $N$ individuals we described in Section \ref{sec:the_market_model}, there is a leading agent, which we call a leader, labeled by $0.$ The game we study is hierarchical in the sense that this leader has a first-mover advantage. The members of the formerly described population are now called followers. Thus, we look for a Stackelberg-Nash type strategy: it optimizes the leader's performance criteria subject to the condition that minor players are following a corresponding Nash equilibrium. Henceforth, we work under the constrained framework, assuming that all of the hypotheses of Theorem \ref{thm:NEStrictLiquidation} hold.

As the agent $0$ changes her strategy, the main technical challenge is that the population of followers will react to it.  Therefore, we take a two-step approach.  Firstly, we assume a leader's turnover rate is given and study the corresponding NE properties in which the population accommodates. Secondly, we insert this equilibrium as a function of the strategy of trader $0$ in the leader's objective criterion, rendering the problem into a single-player optimization.

We assume that the processes and parameters associated with the leader are all adapted to a given filtration $\mathbb{F}^0,$ which we assume complete and continuous. Thus, for a given $q^0_0 \in L^2\left(\Omega,\mathcal{F}^0_0\right),$ we let $\mathbb{A}_{c,0}$ be the set of processes $\nu^0 \in \mathbb{L}^2$ that are $\mathbb{F}^0-$progressively measurable processes, and that satisfy
$$
\int_0^T \nu^0_t\,dt = -q^0_0,\, \mathbb{P}-\text{a.s.}.
$$
Similarly, we define $\mathbb{A}_{c,0}^0$ just as $\mathbb{A}_{c,0}$ but with the identically vanishing random variable in place of $q^0_0.$ We consider processes $\left\{\alpha^0_t \right\}_{0\leqslant t \leqslant T}, \left\{ \kappa^0_t\right\}_{0\leqslant t \leqslant T}$ and $\left\{ \lambda^0_t\right\}_{0\leqslant t \leqslant T},$ alongside constants $\underline{\alpha}^0,\,\underline{\kappa}^0,\,\underline{\lambda}^0,\,\overline{\alpha}^0,\,\overline{\kappa}^0$ and $\overline{\lambda}^0,$ such that
$$
\underline{\alpha}^0 \leqslant \alpha^0 \leqslant \overline{\alpha}^0,\, \underline{\kappa}^0 \leqslant \kappa^0 \leqslant \overline{\kappa}^0 \text{ and } \underline{\lambda}^0 \leqslant \lambda^0 \leqslant \overline{\lambda}^0.
$$

The agent $0$ trades at a rate $\left\{ \nu^0_t \right\}_{0\leqslant t \leqslant T} \in \mathbb{A}_{c,0}.$ In the leader's viewpoint, the mid-price $s^0$ has the dynamics
$$
ds^0_t = \alpha^0_t \nu^0_t\,dt + \alpha_t\mathbb{E}\left[ \frac{1}{N}\sum_{j=1}^N \nu^j_t\,dt \Bigg| \mathcal{F}^0_t\right] + dP^0_t,
$$
with $\left\{ P^0_t,\mathcal{F}^0_t\right\}_{0\leqslant t \leqslant T}$ being a martingale. Her inventory and a cash process evolve as in \eqref{eq:InvDynamics} and \eqref{eq:CashDynamics}, namely,
$$
dq^0_t = \nu^0_t\,dt
$$
and
$$
dc^0_t = -\left( s^0_t + \kappa^0_t \nu^0_t\right)\nu^0_t\,dt.
$$
We particularize the model described in Section \ref{sec:the_market_model} by considering that the drift in \eqref{eq:MidPriceDynamics} takes a particular form, namely,
\begin{equation} \label{eq:DriftLeaderFollowers}
    \mu^i_t = \mathbb{E}\left[ \alpha^0_t \nu^0_t + o_t |\mathcal{F}^i_t \right] \hspace{1.0cm} (i \in \mathcal{N}),
\end{equation}
where $\left\{ o_t\right\}_{0\leqslant t \leqslant T} \in \mathbb{L}^2$ is the exogenous noise traders' order-flow, see \cite{huang2019mean}.

\subsection{Step 1: Given a leader strategy, minor players accommodate to a Nash equilibrium}

Throughout this section, we fix a process $\nu^0 \in \mathbb{L}^2,$ which we assume to be $\mathbb{F}^0-$progressively measurable, and also the order-flow of the noise traders $o \in \mathbb{L}^2.$ The corresponding Nash equilibrium $\boldsymbol{\nu}^* = \left( \nu^{*,1},...,\nu^{*,N} \right)^\intercal$ of the followers, which we assume henceforth to belong to $\mathbb{S}_{c,(N)},$ is the unique solution of the FBSDE \eqref{eq:FBSDEHeterogConstrained} with $\mu^i$ given by \eqref{eq:DriftLeaderFollowers}. Let us write
$$
\boldsymbol{K}\boldsymbol{\nu}^* = \boldsymbol{y} + \boldsymbol{b},
$$
where
\begin{equation} \label{eq:FBSDE_PureArbPart}
    \begin{cases}
        d\boldsymbol{p}_t = \boldsymbol{K}_t^{-1} \boldsymbol{y}_t\,dt,\\
        -d\boldsymbol{y}_t = \boldsymbol{\mathcal{P}}_t\left( \boldsymbol{C}_t \boldsymbol{y}_t\right)\,dt - \left(\boldsymbol{\Sigma}_t - \frac{1}{N}\boldsymbol{\beta}_t \right)\boldsymbol{p}_t\,dt + \boldsymbol{\mathcal{P}}_t\left(\alpha^0_t \nu^0_t \mathbbm{1}_N \right)\,dt - d\boldsymbol{M}^1_t,\\
        \boldsymbol{p}_0 = \boldsymbol{0} \text{ and } \boldsymbol{p}_T = \boldsymbol{0},
    \end{cases}
\end{equation}
and
\begin{equation} \label{eq:FBSDE_ExecPart}
    \begin{cases}
        d\boldsymbol{h}_t = \boldsymbol{K}_t^{-1} \boldsymbol{b}_t\,dt,\\
        -d\boldsymbol{b}_t = \boldsymbol{\mathcal{P}}_t\left( \boldsymbol{C}_t \boldsymbol{b}_t\right)\,dt - \left(\boldsymbol{\Sigma}_t - \frac{1}{N}\boldsymbol{\beta}_t \right)\boldsymbol{h}_t\,dt + \boldsymbol{\mathcal{P}}_t\left(o_t \mathbbm{1}_N \right)\,dt - d\boldsymbol{M}^2_t,\\
        \boldsymbol{h}_0 = \boldsymbol{q}_0 \text{ and } \boldsymbol{h}_T = \boldsymbol{0},
    \end{cases}
\end{equation}
for suitable $\boldsymbol{M}^1,\boldsymbol{M}^2 \in \mathcal{M}_{(N)},$ where $\mathbbm{1}_N$ denotes the $N-$dimensional vector having all entries equal to one. 

Intuitively, we break $\boldsymbol{\nu}$ in two pieces. The first one, $\boldsymbol{K}^{-1}\boldsymbol{y},$ is a pure arbitrage component, whose all entries are round-trip trades. It is through this that followers seek to seize price movements stemming from the leader program. The second part, $\boldsymbol{K}^{-1}\boldsymbol{b},$ concerns the complete execution of the initial portfolio by time $T$ while facing an exogenous excess order-flow $\left\{ o_t \right\}_t$ stemming from noise traders.

We notice that the mapping $\boldsymbol{T} : \nu^0 \mapsto \boldsymbol{y}$ is linear. It is also bounded, as we show next.

\begin{proposition} \label{prop:OperatorTisBdd}
The linear mapping $\boldsymbol{T}$ is bounded, with
$$
\|\boldsymbol{T}\|_{op} := \sup_{\|\nu^0\| = 1} \|\boldsymbol{T}\nu^0\| \leqslant N^{1/2}\frac{2\overline{\alpha}^0\,\overline{\kappa}}{\left(4\underline{\lambda} - \frac{2}{N}\overline{\beta} \right)^{1/2}\epsilon_g},
$$
where $\epsilon_g$ is the gap
$$
\epsilon_g := 2\sqrt{\underline{\kappa}} - \overline{\alpha}\left( 4\underline{\lambda} - \frac{2}{N}\overline{\beta}\right)^{-1/2} > 0.
$$
\end{proposition}
\begin{proof}
It is straightforward to derive from the optimality conditions \eqref{eq:CondnGateuxDerivConstrained} that
\begin{align} \label{eq:Optimatily_NEminors}
    \begin{split}
        \mathbb{E}\Bigg[\int_0^T \Bigg\{- w^i_t y^i_t - 2\lambda^i_t p^i_t \int_0^t w^i_u\,du + \frac{\alpha^i_t}{N} \sum_{j=1}^N \frac{y^j_t}{2\kappa^j_t}\int_0^t w^i_u\,du + \frac{\alpha^i_t p^i_t}{N}w^i_t + \alpha^0_t \nu^0_t\int_0^t w^i_u\,du\Bigg\}\,dt \Bigg] = 0,
    \end{split}
\end{align}
for every $i\in \mathcal{N}$ and $w^i \in \mathbb{A}_{c,i}^0.$ Let us take $w^i = y^i/\left(2\kappa^i\right)$ as a test process in \eqref{eq:Optimatily_NEminors} and sum these equations over $i;$ we estimate, utilizing \eqref{eq:KeyConstants} and Young's inequality with $a,\epsilon > 0,$
\begin{align*}
    \begin{split}
        0 &= \sum_{i=1}^N \mathbb{E}\Bigg[\int_0^T \Bigg\{-\frac{\left(y^i_t\right)^2}{2\kappa^i_t} - \left(2\lambda^i_t - \frac{1}{N}\beta^i_t \right)\left( p^i_t \right)^2 + \frac{\alpha^i_t p^i_t}{N} \sum_{j=1}^N \frac{y^j_t}{2\kappa^j_t} + \alpha^0_t \nu^0_t p^i_t\Bigg\}\,dt \Bigg] \\
        &\leqslant - c_1(a) \|\boldsymbol{y}\|^2 - \left( c_2(1/a) - \epsilon \right) \|\boldsymbol{p}\|^2 + N\frac{\left( \overline{\alpha}^0 \right)^2}{4\epsilon}\|\nu^0\|^2.
    \end{split}
\end{align*}
Upon fixing $\epsilon = c_2(1/a),$ we obtain
\begin{equation} \label{eq:BoundNormOfy_NEminors}
    \|\boldsymbol{y}\|^2 \leqslant N\frac{\left( \overline{\alpha}^0 \right)^2}{4 c_1(a)c_2(1/a)}\|\nu^0\|^2.
\end{equation}
Since
$$
c_1(a) \geqslant \frac{1}{8 \overline{\kappa}^2} \left(4 \underline{\kappa} - a^2 \right) \text{ and } c_2(1/a) \geqslant 2\underline{\lambda} - \frac{1}{N}\overline{\beta} - \frac{\overline{\alpha}^2}{2 a^2},
$$
we write $B:= \left(4\underline{\lambda} - \frac{2}{N}\overline{\beta} \right)^{-1},$ obtaining from \eqref{eq:BoundNormOfy_NEminors} that
\begin{equation} \label{eq:BoundNormOfy_NEminors2}
    \|\boldsymbol{y}\|^2 \leqslant N\frac{4\left( \overline{\alpha}^0 \right)^2 \overline{\kappa}^2 B a^2}{\left(4\underline{\kappa} - a^2 \right) \left( a^2 - \overline{\alpha}^2B \right)}\|\nu^0\|^2.
\end{equation}
Upon taking $a^2 := \sqrt{4\underline{\kappa}\,\overline{\alpha}^2B},$ we minimize the right-hand side of \eqref{eq:BoundNormOfy_NEminors2}, from where the stated estimate follows.  
\end{proof}

It will be convenient for us to introduce the notation
\begin{equation} \label{eq:AvgOfTheNE}
    \overline{\nu}\left(\nu^0\right) := \frac{1}{N}\sum_{i=1}^N \nu^i,
\end{equation}
where $\boldsymbol{\nu} = \left(\nu^1,...,\nu^N \right)^\intercal =  \boldsymbol{K}^{-1}\boldsymbol{T}\nu^0 + \boldsymbol{K}^{-1}\boldsymbol{b}$ is the Nash equilibrium of the followers corresponding to $\nu^0.$ Furthermore, let us write
\begin{equation} \label{eq:BoldL}
    \boldsymbol{L}\nu^0 = \left( L_1\nu^0,...,L_N\nu^0 \right)^\intercal := \boldsymbol{K}^{-1}\boldsymbol{T}\nu^0,
\end{equation}
where we regard $\boldsymbol{K}^{-1}$ as a standard multiplication operator,
$$
\boldsymbol{X} \longmapsto \left\{ \boldsymbol{K}_t^{-1}\boldsymbol{X}_t \right\}_{0\leqslant t \leqslant T}.
$$
We also introduce the continuous linear mapping $\nu^0 \mapsto \overline{L}\nu^0$ by
\begin{equation} \label{eq:Lbar}
    \overline{L}\nu^0 = \frac{1}{N}\sum_{i=1}^{N} L_i \nu^0.
\end{equation}
From Proposition \ref{prop:OperatorTisBdd}, we can estimate the norm of $\overline{L}$ uniformly in $N.$
\begin{corollary} \label{cor:EstimateOnNormOfLbar}
The following estimate for the operator norm of $\overline{L}$ holds:
$$
\left\|\overline{L}\right\|_{op} \leqslant \frac{\overline{\alpha}^0\,\overline{\kappa}}{\left[\underline{\kappa}\left(2\underline{\lambda} - \frac{1}{N}\overline{\beta} \right) \right]^{1/2}\epsilon_g}.
$$
\end{corollary}
\begin{proof}
It is straightforward to prove relation
$$
\left\| \overline{L} \right\|_{op}^2 \leqslant \frac{1}{2\underline{\kappa} N}\left\| \boldsymbol{T}\right\|_{op}^2.
$$
This inequality, alongside Proposition \ref{prop:OperatorTisBdd}, imply the result.  
\end{proof}

We conclude this subsection by investigating the adjoint operator of $\overline{L}.$ From \eqref{eq:BoldL} and \eqref{eq:Lbar}, we see that this issue is intricately related to the one corresponding to $\boldsymbol{T},$ leading us to the following Lemma.
\begin{lemma} \label{lem:AdjOfT}
Given $\boldsymbol{g} \in \mathbb{A}_{c,(N)}^0,$ let us assume that the FBSDE 
\begin{equation*} 
    \begin{cases}
        -d\boldsymbol{\varphi}_t = \left(\boldsymbol{\Sigma}_t - \frac{1}{N}\boldsymbol{\beta}_t \right)\boldsymbol{\psi}_t\,dt - d\boldsymbol{\Gamma}_t,\\
        d\boldsymbol{\psi}_t = -\boldsymbol{K}^{-1}_t\boldsymbol{\varphi}_t\,dt +\boldsymbol{\mathcal{P}}_t\left( \boldsymbol{C}_t^\intercal \boldsymbol{\psi}_t\right)\,dt + \boldsymbol{g}_t\,dt,\\
        \boldsymbol{\psi}_0 = \boldsymbol{0} \text{ and } \boldsymbol{\psi}_T = \boldsymbol{0},
    \end{cases}
\end{equation*}
admits a solution $\left( \boldsymbol{\varphi},\boldsymbol{\psi},\boldsymbol{\Gamma}\right) \in \mathbb{A}_{(N)}\times \mathbb{A}_{(N)} \times \mathcal{M}_{(N)}.$ Then,
$$
\left( \boldsymbol{T}^*\boldsymbol{g} \right)_t = \alpha^0_t\mathbb{E}\left[\mathbbm{1}_N^\intercal\boldsymbol{\psi}_t |\mathcal{F}^0_t \right] \hspace{1.0cm} \left(t \in \left[0,T\right]\right).
$$
\end{lemma}
\begin{proof}
Let us consider $\left(\boldsymbol{p},\boldsymbol{y},\boldsymbol{M}^1\right)$ as in \eqref{eq:FBSDE_PureArbPart}. Using It\^o's formula, we compute
\begin{align*}
    0 &= \mathbb{E}\left[ \left(\boldsymbol{p}_t\cdot \boldsymbol{\varphi}_t + \boldsymbol{y}_t\cdot \boldsymbol{\psi}_t \right)|_{t=0}^{t=T} \right] \\
    &= \mathbb{E}\left[\int_0^T \left\{-\alpha^0_t\mathbb{E}\left[ \mathbbm{1}_N^\intercal \boldsymbol{\psi}_t |\mathcal{F}^0_t\right] \nu^0_t + \boldsymbol{y}\cdot \boldsymbol{g} \right\}\,dt \right],
\end{align*}
whence the result follows.  
\end{proof}

With the aid of Lemma \ref{lem:AdjOfT}, we provide a useful description of $\overline{L}^*.$
\begin{proposition} \label{prop:AdjOfLbar}
Given $g \in \mathbb{L}^2,$ let us assume that the FBSDE
\begin{equation} \label{eq:AdjointFSBDE}
    \begin{cases}
        -d\boldsymbol{\varphi}_t = \left(\boldsymbol{\Sigma}_t - \frac{1}{N}\boldsymbol{\beta}_t \right)\boldsymbol{\psi}_t\,dt - d\boldsymbol{\Gamma}_t,\\
        d\boldsymbol{\psi}_t = -\boldsymbol{K}^{-1}_t\boldsymbol{\varphi}_t\,dt +\boldsymbol{\mathcal{P}}_t\left( \boldsymbol{C}_t^\intercal \boldsymbol{\psi}_t\right)\,dt + \boldsymbol{\mathcal{P}}_t\left( \frac{1}{N}g_t \boldsymbol{K}_t^{-1}\mathbbm{1}_N \right)\,dt ,\\
        \boldsymbol{\psi}_0 = \boldsymbol{0} \text{ and } \boldsymbol{\psi}_T = \boldsymbol{0},
    \end{cases}
\end{equation}
admits a solution $\left(\boldsymbol{\varphi},\boldsymbol{\psi},\boldsymbol{\Gamma}\right) \in \mathbb{A}_{(N)} \times \mathbb{A}_{(N)} \times \mathcal{M}_{(N)}.$ Then,
\begin{equation} \label{eq:LbarStar}
    \left( \overline{L}^* g \right)_t = \alpha^0_t\mathbb{E}\left[\mathbbm{1}_N^\intercal\boldsymbol{\psi}_t |\mathcal{F}^0_t \right] \hspace{1.0cm} \left(t \in \left[0,T\right]\right).
\end{equation}
\end{proposition}
\begin{proof}
For $i \in \mathcal{N},$ we consider the canonical projection and injection, respectively:
$$
\pi_i : \boldsymbol{X} = (X^1,...,X^N)^\intercal \in \mathbb{A}_{(N)} \mapsto X^i \in \mathbb{L}^2;
$$
$$
\iota_i : X \in \mathbb{L}^2 \mapsto \left(0,...,0,\mathcal{P}^i\left(X\right),0,...,0\right)^\intercal \in \mathbb{A}_{(N)}.
$$
It is clear that $\left( \pi_i \right)^* = \iota_i.$ We notice that
$$
\overline{L} = \left( \frac{1}{N}\sum_{i=1}^N \pi_i \right) \boldsymbol{L} = \left( \frac{1}{N}\sum_{i=1}^N \pi^i \right) \boldsymbol{K}^{-1}\boldsymbol{T}.
$$ 
Therefore,
\begin{equation} \label{eq:LBarStarInTermsOfTstar}
    \overline{L}^* = \boldsymbol{T}^*\boldsymbol{K}^{-1} \left(\frac{1}{N}\sum_{i=1}^N \iota_i \right),
\end{equation}
since $\boldsymbol{K}$ is symmetric. In view of Lemma \ref{lem:AdjOfT}, identity \eqref{eq:LBarStarInTermsOfTstar} implies \eqref{eq:LbarStar}.  
\end{proof}

\subsection{Step 2: The leader solves her optimization problem}

Following the same ideas of \eqref{eq:main_payoff}, we define the performance criteria of the leader in the sequel.
\begin{definition} \label{defn:MajorPlayerObjFunc}
The objective functional $J_0 : \mathbb{A}_{c,0} \rightarrow \mathbb{R}$ of player $0$ is defined as
\begin{align*}
    J_0(\nu^0) := \mathbb{E}\left[\int_0^T \left\{-\kappa^0_t \left( \nu^0_t \right)^2 - \lambda^0_t\left( q^0_t \right)^2 + q^0_t\left[ \alpha^0_t\nu^0_t + \alpha^0_t \overline{\nu}\left( \nu^0 \right)_t + o_t \right] \right\}\,dt \right].
\end{align*}
\end{definition}

\begin{lemma} \label{lem:DiffProptsOfJ0}
(a) Given $\nu^0 \in \mathbb{A}_{c,0}$ and $w^0 \in \mathbb{A}_{c,0}^0,$ the first order G\^ateaux derivative
$$
\left\langle J_0^\prime(\nu^0), w^0 \right\rangle = \lim_{\epsilon \rightarrow 0} \frac{J_0\left( \nu^0 +\epsilon w^0\right) - J_0\left(\nu^0\right)}{\epsilon}
$$
exists.

(b) For $\nu^0$ and $w^0$ as in (a), the second order G\^ateaux derivative
$$
\left\langle J_0^{\prime\prime}(\nu^0), (w^0,w^0) \right\rangle = \lim_{\epsilon \rightarrow 0} \frac{1}{\epsilon} \left\langle J_0^{\prime}\left( \nu^0 +\epsilon w^0\right) - J_0^{\prime}\left(\nu^0\right),w^0\right\rangle
$$
exists and is independent of $\nu^0.$ Consequently, for $\nu^0,\widetilde{\nu}^0 \in \mathbb{A}_{c,0},$ we have the identity
\begin{equation} \label{eq:J0IsQuadratic}
    J_0(\nu^0) = J_0(\widetilde{\nu}^0) + \left\langle J_0^\prime(\widetilde{\nu}^0), \nu^0 - \widetilde{\nu}^0 \right\rangle + \frac{1}{2}\left\langle J_0^{\prime \prime}(\widetilde{\nu}^0), \left( \nu^0 - \widetilde{\nu}^0, \nu^0 - \widetilde{\nu}^0\right) \right\rangle.
\end{equation}

(c) If 
\begin{equation} \label{eq:UltraWeakInteraction}
    \left(\overline{\alpha}^0\right)^2 < 2\left(\left\|\overline{L}\right\|_{op}^2 + 1\right)^{-1}\underline{\lambda}^0\,\underline{\kappa}^0,
\end{equation}
then there exists $c>0,$ depending only on model parameters, such that
\begin{equation} \label{eq:NegativeDeftnessOfD2J0}
    \left\langle J_0^{\prime\prime}(\nu^0),(w^0,w^0)\right\rangle \leqslant - c\|w^0\|^2.
\end{equation}
\end{lemma}
\begin{remark}
By Corollary \ref{cor:EstimateOnNormOfLbar}, we notice that the norm of the operator $\overline{L}$ appearing in the weak interaction assumption \eqref{eq:UltraWeakInteraction} made in ($c$) is uniformly bounded on the population size $N.$ Therefore, this result is meaningful even for large populations. Moreover, if we utilize in \eqref{eq:UltraWeakInteraction} the estimate provided in Corollary \ref{cor:EstimateOnNormOfLbar}, we derive an explicit condition for the result to hold solely in terms of upper and lower bounds on the model parameters.
\end{remark}
\begin{proof}
($a$) We can establish this item by proceeding just as we did in Lemma \ref{lem:GateauxDerivWellDefnd}. In this way, we obtain the formula
\begin{align} \label{eq:FormulaDJ0}
  \begin{split}
    \left\langle J^\prime(\nu^0),w^0 \right\rangle &= \mathbb{E}\Bigg[ \int_0^T \Bigg\{-2\kappa^0_t \nu^0_t w^0_t - 2\lambda^0_t q^0_t \int_0^t w^0_u\,du + \left(\alpha^0_t \nu^0_t + \alpha^0_t\overline{\nu}\left(\nu^0\right)_t + o_t \right)\int_0^t w^0_u\,du \\
    &\hspace{7.25cm}+ q^0_t \left[\alpha^0_t w^0_t + \alpha^0_t \left(\overline{L} w^0 \right)_t \right] \Bigg\}\,dt \Bigg] \\
    &= \mathbb{E}\Bigg[\int_0^T w^0_t\Bigg\{-2\kappa^0_t \nu^0_t -2\int_t^T \lambda^0_u q^0_u\,du  + \int_t^T\left(\alpha^0_u \nu^0_u + \alpha^0_u\overline{\nu}\left(\nu^0\right)_u + o_u \right)\,du \\
    &\hspace{7.25cm} + \alpha^0_t q^0_t + \left( \overline{L}^* \left(\alpha^0 q^0 \right) \right)_t\Bigg\}\,dt \Bigg].
  \end{split}
\end{align}

($b$) We notice from \eqref{eq:FormulaDJ0} that, for each $w^0 \in \mathbb{A}_{c,0}^0,$ the mapping $\nu^0 \mapsto \left\langle J^\prime\left(\nu^0\right),w^0\right\rangle$ is linear and continuous. Furthermore, we see that the second-order derivative does exist and that it is equal to
\begin{align*}
    \left\langle J_0^{\prime\prime}\left(\nu^0\right),\left(w^0,w^0\right)\right\rangle = \mathbb{E}\Bigg[\int_0^T \Bigg\{ -2\kappa^0_t \left(w^0\right)^2 -2\lambda^0_t \left(\int_0^t w^0_u\,du\right)^2 + 2\alpha^0_t\int_0^t w^0_u\,du \left[w^0_t + \left( \overline{L} w^0 \right)_t \right] \Bigg\}\,dt \Bigg],
\end{align*}
which is independent of $\nu^0.$ Given $\nu^0,\widetilde{\nu}^0 \in \mathbb{A}_{c,0},$ we have $w^0 := \nu^0 -\widetilde{\nu}^0 \in \mathbb{A}^0_{c,0};$ hence, we can derive \eqref{eq:J0IsQuadratic} as follows:
\begin{align*}
    J_0(\nu^0) - J_0(\widetilde{\nu}^0) &= \int_0^1 \left\langle J_0^\prime\left(\theta_1 \nu^0 + \left(1-\theta_1\right)\widetilde{\nu}^0 \right),w^0\right\rangle \,d\theta_1 \\
    &= \left\langle J_0^\prime\left(\widetilde{\nu}^0 \right),w^0\right\rangle + \int_0^1 \left\langle J_0^\prime\left(\theta_1 \nu^0 + \left(1-\theta_1\right)\widetilde{\nu}^0 \right) - J_0^\prime\left(\widetilde{\nu}^0 \right),w^0\right\rangle \,d\theta_1 \\
    &= \left\langle J_0^\prime\left(\widetilde{\nu}^0 \right),w^0\right\rangle + \int_0^1\int_0^1 \left\langle  J_0^{\prime\prime}\left(\theta_1\theta_2 \nu^0 + \left(1-\theta_1\theta_2\right)\widetilde{\nu}^0 \right),\left(w^0,w^0\right)\right\rangle \theta_1\,d\theta_2\, d\theta_1 \\
    &= \left\langle J_0^\prime\left(\widetilde{\nu}^0 \right),w^0\right\rangle + \int_0^1\int_0^1 \theta_1 \,d\theta_2\,d\theta_1 \left\langle J_0^{\prime\prime}\left(\widetilde{\nu}^0 \right),\left(w^0,w^0\right)\right\rangle \\
    &= \left\langle J_0^\prime\left(\widetilde{\nu}^0 \right),w^0\right\rangle + \frac{1}{2} \left\langle J_0^{\prime\prime}\left(\widetilde{\nu}^0 \right),\left(w^0,w^0\right)\right\rangle.
\end{align*}

($c$) Employing Young's inequality with $a>0,$ we infer
\begin{align*}
    \big\langle J_0^{\prime\prime}\left(\nu^0\right),\left(w^0,w^0\right)\big\rangle \leqslant -\left[2\underline{\kappa}^0 - a^2\left(1 + \left\|\overline{L}\right\|_{op}^2 \right) \right]\left\|w^0\right\|^2 - \left[2\underline{\lambda}^0 - \frac{2\left(\overline{\alpha}^0\right)^2}{a^2} \right]\left\| \int_0^\cdot w_u\,du \right\|^2.
\end{align*} 
We take $a$ subject to
$$
\frac{\left( \overline{\alpha}^0 \right)^2}{\underline{\lambda}^0} < a^2 < \frac{2\underline{\kappa}^0}{1+\left\|\overline{L}\right\|_{op}^2},
$$
whence \eqref{eq:NegativeDeftnessOfD2J0} follows.  
\end{proof}

\begin{corollary} \label{cor:J0StrictConcAdCoerc}
Let us assume that \eqref{eq:UltraWeakInteraction} is in force. Then, the functional $J_0$ is strictly concave and coercive.
\end{corollary} 
\begin{proof}
We fix $\widetilde{\nu}^0$ and notice that the mapping
$$
\nu^0 \mapsto J_0\left(\widetilde{\nu}^0\right) + \left\langle J_0^\prime\left(\widetilde{\nu}^0\right),\nu-\widetilde{\nu}^0\right\rangle + \frac{1}{2}\left\langle J_0^{\prime\prime}\left(\widetilde{\nu}^0\right),\left(\nu^0 - \widetilde{\nu}^0,\nu^0 - \widetilde{\nu}^0\right)\right\rangle
$$
is concave; hence, the concavity of $J_0$ follows from its representation proved in Lemma \ref{lem:DiffProptsOfJ0} ($b$), whereas strict concavity follows from Lemma \ref{lem:DiffProptsOfJ0} ($c$). Similarly, these results give
$$
J_0\left(\nu^0\right) \leqslant J\left(\widetilde{\nu}^0\right) + \left\|J^\prime\left(\widetilde{\nu}^0\right)\right\|\left\|\nu^0 - \widetilde{\nu}^0\right\| - c\left\|\nu^0 - \widetilde{\nu}^0\right\|^2 \rightarrow -\infty,
$$
as $\left\|\nu^0\right\| \rightarrow \infty,$ where $J^\prime\left(\widetilde{\nu}^0\right)$ is given by
\begin{align*}
    J^\prime\left(\widetilde{\nu}^0\right)_t = -2\kappa^0_t \widetilde{\nu}^0_t + \alpha^0_t \widetilde{q}^0_t + \mathbb{E}\left[ -2 \int_t^T \lambda^0_u \widetilde{q}^0_u\,du + \int_t^T\left(\alpha^0_u \widetilde{\nu}^0_u + \alpha^0_u\overline{\nu}\left(\widetilde{\nu}^0\right)_u + o_u \right)\,du  + \left( \overline{L}^* \left(\alpha^0 \widetilde{q}^0 \right) \right)_t \bigg| \mathcal{F}^0_t\right],
\end{align*}
if $\widetilde{q}^0_t = q^0_0 + \int_0^t \widetilde{\nu}^0_u\,du$ and $t \in \left[0,T\right].$ Therefore, $J_0$ is coercive.  
\end{proof}

Putting together Corollary \ref{cor:J0StrictConcAdCoerc}, identity \eqref{eq:FormulaDJ0}, and the reasoning of Lemma \ref{lem:TestingAgainstA0}, we conclude:
\begin{theorem} \label{thm:CharactMajorStrat}
Under condition \eqref{eq:UltraWeakInteraction}, there exists a unique solution $\nu^{*,0}$ of the optimization problem
\begin{equation} \label{eq:MajorMaximazation}
    \nu^{*,0} = \argmax_{\nu^0 \in \mathbb{A}_{c,0} } J_0\left(\nu^0\right).
\end{equation}
Moreover, we characterize it as the solution of the system
\begin{equation} \label{eq:AbstractFBSDEMajor}
    \begin{cases}
        q^{*,0}_t = q^0_0 + \int_0^t \nu^{*,0}_u\,du,\\
        2\kappa^0_t \nu^{*,0}_t = 2\int_0^t \lambda^0_u q^{*,0}_u du  - \int_0^t \alpha^0_u \nu^{*,0}_u\,du -\, \int_0^t\mathbb{E}\left[ \alpha^0_u\overline{\nu}\left(\nu^{*,0}\right)_u + o_u | \mathcal{F}^0_u\right] \,du +\alpha^0_t q^{*,0}_t \\
        \hspace{9.025cm}+\, \left( \overline{L}^* \left(\alpha^0 q^{*,0} \right) \right)_t + M^{*,0}_t,\\
        q^{*,0}_0 = q^0_0 \text{ and } q^{*,0}_T = 0,
    \end{cases}
\end{equation}
for some square-integrable martingale $\left\{ M^{*,0}_t,\mathcal{F}^0_t \right\}_{0\leqslant t \leqslant T}.$
\end{theorem}

Consequently, if we condense \eqref{eq:FBSDE_PureArbPart}, Proposition \ref{prop:AdjOfLbar}, and Theorem \ref{thm:CharactMajorStrat}, we can describe the optimal leader strategy $\nu^{*,0}$ and the population's NE $\boldsymbol{\nu}$ via a coupled FBSDE system including an adjoint state.
\begin{corollary} \label{cor:adjoinedFBSDE}
Let us assume that $\left(q^0,\boldsymbol{p},\boldsymbol{\varphi},x^0,\boldsymbol{x},\boldsymbol{\psi},M^0,\boldsymbol{M}^1,\boldsymbol{\Gamma}\right)$ solves the adjoined FBSDE
$$
\begin{cases}
        dq^0_t = \frac{x^0_t}{2\kappa^0_t}\,dt,\\
        d\boldsymbol{p}_t = \boldsymbol{K}_t^{-1} \boldsymbol{y}_t\,dt,\\
        -d\boldsymbol{\varphi}_t = \left(\boldsymbol{\Sigma}_t - \frac{1}{N}\boldsymbol{\beta}_t \right)\boldsymbol{\psi}_t\,dt -  d\boldsymbol{\Gamma}_t,\\
        -dx^0_t = -2\lambda^0_t q^0_t + \frac{\alpha^0_t}{2\kappa^0_t}x^0_t\,dt + \frac{\alpha^0_t}{N}\mathbb{E}\left[\mathbbm{1}_N^\intercal\boldsymbol{K}_t^{-1}\left( \boldsymbol{y}_t + \boldsymbol{b}_t\right) | \mathcal{F}^0_t \right]\,dt \\
        \hspace{2.5cm}+\,\mathbb{E}\left[o_t | \mathcal{F}^0_t \right]\,dt - d \left( \alpha^0_t q^0_t + \alpha^0_t\mathbb{E}\left[ \mathbbm{1}_N^\intercal\boldsymbol{\psi}_t |\mathcal{F}^0_t\right] \right) - dM^0_t,\\
        -d\boldsymbol{y}_t = \boldsymbol{\mathcal{P}}_t\left( \boldsymbol{C}_t \boldsymbol{y}_t\right)\,dt - \left(\boldsymbol{\Sigma}_t - \frac{1}{N}\boldsymbol{\beta}_t \right)\boldsymbol{p}_t\,dt + \boldsymbol{\mathcal{P}}_t\left(\alpha^0_t \nu^0_t \mathbbm{1}_N \right)\,dt - d\boldsymbol{M}^1_t,\\
        d\boldsymbol{\psi}_t = -\boldsymbol{K}^{-1}_t\boldsymbol{\varphi}_t\,dt +\boldsymbol{\mathcal{P}}_t\left( \boldsymbol{C}_t^\intercal \boldsymbol{\psi}_t\right)\,dt + \boldsymbol{\mathcal{P}}_t\left( \frac{\alpha^0_t}{N}q^0_t \boldsymbol{K}_t^{-1}\mathbbm{1}_N \right)\,dt,\\
        q^0_0 \text{ given in } L^2\left(\Omega,\mathcal{F}^0_0\right),\, q^0_T=0,\\
        \boldsymbol{p}_0 = \boldsymbol{0},\, \boldsymbol{p}_T = \boldsymbol{0},\\
        \boldsymbol{\psi}_0 = \boldsymbol{0} \text{ and } \boldsymbol{\psi}_T = \boldsymbol{0},
\end{cases}
$$
where $\boldsymbol{M}^1,\boldsymbol{\Gamma} \in \mathcal{M}_{(N)},$ and $\left\{ M^0_t,\mathcal{F}^0_t\right\}_{0\leqslant t \leqslant T}$ is a square-integrable martingale. Then, the process $\nu^0 := \left(2\kappa^0\right)^{-1}x^0$ is the optimal control for the leader, whereas $\boldsymbol{\nu} := \boldsymbol{K}^{-1}\left( \boldsymbol{y} + \boldsymbol{b} \right)$ is a Nash equilibrium of the followers.
\end{corollary}

\subsection{Constant model parameters}

Let us assume all model parameters and initial datum are constant. We also suppose that, among the population of followers, parameters are homogeneous, i.e.,
$$
\alpha^i \equiv \alpha,\, \kappa^i \equiv \kappa,\, \lambda^i \equiv \lambda \text{ and } \beta^i \equiv 0.
$$
For simplicity, we assume an absent noise traders' order-flow $o \equiv 0,$ and that the initial inventory holdings of the followers vanish in average. It is straightforward to derive that these simplifying assumptions imply $\mathbbm{1}_N^\intercal \boldsymbol{b} \equiv 0$ (cf. equation \eqref{eq:FBSDE_ExecPart}). 

Let us assume $\left(q^0,\boldsymbol{p},\boldsymbol{\varphi},x^0,\boldsymbol{y},\boldsymbol{\psi} \right)$ solves the adjoined FBSDE system presented in Corollary \ref{cor:adjoinedFBSDE}. We introduce the dependent variables
$$
\overline{p} := \frac{1}{N}\mathbbm{1}_N^\intercal \boldsymbol{p},\, \overline{y} := \frac{1}{N}\mathbbm{1}_N^\intercal \boldsymbol{y},\,\overline{\varphi} := \mathbbm{1}_N^\intercal\boldsymbol{\varphi} \text{ and } \overline{\psi} := \mathbbm{1}_N^\intercal\boldsymbol{\psi}.
$$
We remark that $2\kappa\overline{y}$ is the average rate of trading of the minor agents in the Nash equilibrium, and likewise $\overline{p}$ is their average inventory holdings. In the present deterministic setting, we derive from Corollary \ref{cor:adjoinedFBSDE} that
$$
\boldsymbol{E} = \left( q^0, \overline{p},\overline{\varphi} \right)^\intercal \text{ and } \boldsymbol{F} = \left(x^0,\overline{y},\overline{\psi} \right)^\intercal
$$
solve the ODE system
$$
\begin{cases}
\dot{\boldsymbol{E}} = \boldsymbol{H}_1 \boldsymbol{F},\\
\dot{\boldsymbol{F}} = \boldsymbol{H}_2 \boldsymbol{E} + \boldsymbol{H}_3 \boldsymbol{F},
\end{cases}
$$
with boundary conditions (BC)
$$
E_1(0) = q^0_0,\, E_1(T) = 0,\, E_2(0) =0,\, E_2(T) = 0,\, F_3(0) = 0 \text{ and } F_3(T) = 0,
$$
where $\boldsymbol{E} = \left( E_1,E_2,E_3\right)^\intercal,\, \boldsymbol{F} = \left( F_1,F_2,F_3\right)^\intercal,\,\boldsymbol{e}_1 = \left( 1,0,0 \right)^\intercal,$ and the matrices $\boldsymbol{H}_1, \boldsymbol{H}_2$ and $\boldsymbol{H}_3$ are given by
$$
\boldsymbol{H}_1 = \begin{bmatrix}
\frac{1}{2\kappa^0} & 0 & 0 \\
0 & \frac{1}{2\kappa} & 0 \\
0 & 0 & -2\lambda
\end{bmatrix},
$$
$$
\boldsymbol{H}_2 = \begin{bmatrix}
2\lambda^0 + \frac{\left( \alpha^0 \right)^2}{2\kappa} & 0 & - \frac{\alpha^0}{2\kappa} \\
0 & 2\lambda & 0 \\
\frac{\alpha^0}{2\kappa} & 0 & -\frac{1}{2\kappa}
\end{bmatrix},
$$
and
$$
\boldsymbol{H}_3 = \begin{bmatrix}
0 & -\frac{\alpha^0}{2\kappa} & \frac{\alpha^0 \alpha}{2\kappa}\left( 1 - \frac{1}{N} \right) \\
-\frac{\alpha^0}{2\kappa^0} & - \frac{\alpha}{2\kappa}\left( 1 - \frac{1}{N} \right) & 0 \\
0 & 0 & \frac{\alpha}{2\kappa}\left( 1 - \frac{1}{N} \right)
\end{bmatrix}.
$$
We can express everything in terms of $\boldsymbol{E}$ only, as in Theorem \ref{thm:ConstParamsHomog_Unconstrained}; it solves the ODE system
\begin{equation} \label{eq:ODE_HierarchicalMkt}
    \ddot{\boldsymbol{E}} = \boldsymbol{H}_1\boldsymbol{H}_3 \boldsymbol{H}_1^{-1} \dot{\boldsymbol{E}} + \boldsymbol{H}_1\boldsymbol{H}_2 \boldsymbol{E}, \, \text{ in } \left]0,T\right[,
\end{equation}
with BC
\begin{equation} \label{eq:BC_HierarchicalMkt}
    \begin{cases}
        E_1(0) = q^0_0,\, E_1(T) = 0,\\
        E_2(0) =0,\, E_2(T) = 0,\\
        \dot{E}_3(0) = 0 \text{ and } \dot{E}_3(T) = 0.
    \end{cases}
\end{equation}

The last main result of this paper concerns the existence and uniqueness of a classical solution of the ODE \eqref{eq:ODE_HierarchicalMkt} with BC \eqref{eq:BC_HierarchicalMkt}. In preparation to prove it, we observe that $\boldsymbol{E}$ solves \eqref{eq:ODE_HierarchicalMkt} and \eqref{eq:BC_HierarchicalMkt} if, and only if, 
$$
\boldsymbol{g} \equiv \left(g_1,g_2,g_3\right)^\intercal :=\boldsymbol{H}_1^{-1}\left(  \boldsymbol{E} - q^0_0\left(1-\frac{t}{T}\right)\boldsymbol{e}_1 \right)
$$
solves
\begin{equation} \label{eq:ODE_HierarchicalMkt_translated}
    \ddot{\boldsymbol{g}} = \boldsymbol{H}_3 \dot{\boldsymbol{g}} + \boldsymbol{H}_2\boldsymbol{H}_1 \boldsymbol{g} + \boldsymbol{h}, \, \text{ in } \left]0,T\right[,
\end{equation}
with homogeneous BC
\begin{equation} \label{eq:BC_HierarchicalMkt_homog}
    \begin{cases}
        g_1(0) = 0,\, g_1(T) = 0,\\
        g_2(0) =0,\, g_2(T) = 0,\\
        \dot{g}_3(0) = 0 \text{ and } \dot{g}_3(T) = 0,
    \end{cases}
\end{equation}
where
$$
\boldsymbol{h}(t) := -\frac{q^0_0}{T} \boldsymbol{H}_3 \boldsymbol{H}_1^{-1}\boldsymbol{e}_1 +q^0_0\left(1-\frac{t}{T}\right) \boldsymbol{H}_2\boldsymbol{e}_1.
$$
We introduce the Hilbert space
\begin{align*}
    \mathcal{H} := \bigl\{ \boldsymbol{v} = \left(v_1,v_2,v_3\right)^\intercal :\,& \boldsymbol{v},\,\dot{\boldsymbol{v}} \in L^2\left(0,T\right), \text{ and } v_1(0) = v_1(T) = v_2(0) = v_2(T) = 0 \bigr\},
\end{align*}
with the inner product
$$
\left( \boldsymbol{v}, \boldsymbol{w}\right)_{\mathcal{H}} := \int_0^T \left( \dot{\boldsymbol{v}} \cdot \dot{\boldsymbol{w}} + \boldsymbol{v}\cdot \boldsymbol{w} \right)\,dt. 
$$

\begin{lemma} \label{lem:Hierarchical_BilinearFormPropts}
Let us assume that
\begin{equation} \label{eq:WeakInteractionODE_Hierarchy}
    c^* < 16\lambda^*,
\end{equation}
where
$$
c^* := \left( \alpha^0 \right)^2\left( \frac{1}{\kappa^0} - \frac{1}{\kappa} \right)^2 + \left(\frac{\alpha}{\kappa} \right)^2 \left[ 1 + \left( \alpha^0\right)^2 \right] \left( 1 - \frac{1}{N} \right)^2
$$
and
$$
\lambda^* := \frac{\lambda}{\kappa} \wedge \left[ \frac{1}{2}\left(\frac{\lambda^0}{\kappa^0} + \frac{\left( \alpha^0 \right)^2}{4\kappa \kappa^0} + \frac{\lambda}{\kappa} - \sqrt{\left(\frac{\lambda^0}{\kappa^0} + \frac{\left( \alpha^0 \right)^2}{4\kappa \kappa^0} - \frac{\lambda}{\kappa} \right)^2 + \frac{\left(\alpha^0\right)^2\lambda}{\kappa^2\kappa^0}} \right) \right].
$$
Then, the bilinear form $b : \mathcal{H}\times \mathcal{H} \rightarrow \mathbb{R}$ defined as
\begin{equation} \label{eq:TheBilinearForm_Hierarchy}
    b\left( \boldsymbol{v},\boldsymbol{w}\right) := \int_0^T \left( \dot{\boldsymbol{v}}\cdot \dot{\boldsymbol{w}} + \boldsymbol{v} \cdot \boldsymbol{H}_3 \dot{\boldsymbol{w}} + \boldsymbol{v} \cdot \boldsymbol{H}_2\boldsymbol{H}_1 \boldsymbol{w} \right)\,dt
\end{equation}
is coercive, i.e.,
$$
b\left( \boldsymbol{v},\boldsymbol{v}\right) \geqslant C \|\boldsymbol{v}\|_{\mathcal{H}}^2.
$$
\end{lemma}
\begin{proof}
On the one hand, let us observe that the matrix $\boldsymbol{H}_2\boldsymbol{H}_1$ is strictly positive definite, and its smallest eigenvalue is precisely $\lambda^*.$ On the other hand, it is straightforward to see that
$$
\left| \int_0^T \boldsymbol{v} \cdot \boldsymbol{H}_3 \dot{\boldsymbol{v}} \, dt \right| \leqslant \epsilon\left| \dot{\boldsymbol{v}}\right|^2 + \frac{c^*}{16\epsilon} \left| \boldsymbol{v}\right|^2,
$$
for every $\boldsymbol{v} \in \mathcal{H}.$ Therefore,
$$
b\left( \boldsymbol{v},\boldsymbol{v}\right) \geqslant \left(1-\epsilon\right)\left|\dot{\boldsymbol{v}}\right|^2 + \left(\lambda^* - \frac{c^*}{16\epsilon} \right)\left| \boldsymbol{v} \right|^2.
$$
Choosing $\epsilon \in \left]0,1\right[$ suitably, i.e.,
$$
\frac{c^*}{16\lambda^*} < \epsilon < 1,
$$
we establish the result.  
\end{proof}

\begin{theorem}
Under the weak interaction assumption \eqref{eq:WeakInteractionODE_Hierarchy} of Lemma \ref{lem:Hierarchical_BilinearFormPropts}, the ODE \eqref{eq:ODE_HierarchicalMkt_translated} with BC \eqref{eq:BC_HierarchicalMkt_homog} has a unique twice continuously differentiable solution.
\end{theorem}
\begin{proof}
Since the bilinear form $b$ we defined in \eqref{eq:TheBilinearForm_Hierarchy} is continuous, as well as the linear functional
$$
\boldsymbol{v} \in \mathcal{H} \mapsto -\int_0^T \boldsymbol{h}\cdot \boldsymbol{v}\,dt \in \mathbb{R},
$$
it follows from the Lax-Milgram Lemma, see \cite[Corollary 5.8]{brezis2010functional}, that there exists a unique $\boldsymbol{g} \in \mathcal{H}$ such that
\begin{equation} \label{eq:VariationalProblem}
    b\left(\boldsymbol{g},\boldsymbol{v}\right) = -\int_0^T \boldsymbol{h}\cdot \boldsymbol{v}\,dt,
\end{equation}
for every $\boldsymbol{v} \in \mathcal{H}.$ We claim that $\ddot{\boldsymbol{g}} \in L^2\left(0,T\right).$ In effect, for every infinitely differentiable $\boldsymbol{v},$ compactly supported within $\left]0,T\right[,$ we have the subsequent relations in the distributional sense:
\begin{align*}
    \left\langle \ddot{\boldsymbol{g}},\boldsymbol{v} \right\rangle &= - \int_0^T \dot{\boldsymbol{g}}\cdot \dot{\boldsymbol{v}}\,dt \\
    &=\int_0^T \left( \boldsymbol{H}_3 \dot{\boldsymbol{g}} + \boldsymbol{H}_2\boldsymbol{H}_1 \boldsymbol{g} + \boldsymbol{h} \right)\cdot \boldsymbol{v}\,dt.
\end{align*}
Thus, our assertion is valid, and we further have that \eqref{eq:ODE_HierarchicalMkt_translated} is valid at a.e. $t\in \left[0,T\right].$ 

The memberships of the functions $\boldsymbol{g},\, \dot{\boldsymbol{g}}$ and $\ddot{\boldsymbol{g}}$ in $L^2\left(0,T\right)$ imply $\boldsymbol{g},\,\dot{\boldsymbol{g}} \in C\left( \left[0,T\right]\right),$ see \cite[Chapter 5, Section 5.9, Theorem 2 (i)]{evans10pde}, whence the right-hand side of \eqref{eq:ODE_HierarchicalMkt_translated} is continuous. From \cite[Chapter 6, Section 6.10]{lieb2001analysis}, we conclude that $\boldsymbol{g}$ is of class $C^2$ (i.e., twice continuously differentiable). 

Since $\boldsymbol{g} = \left(g_1,g_2,g_3\right)^\intercal \in \mathcal{H},$ the adequate boundary conditions for $g_1$ and $g_2$ hold. As for $g_3,$ let us test equation \eqref{eq:ODE_HierarchicalMkt_translated} against a $\boldsymbol{v} \in \mathcal{H}$ to obtain:
\begin{align*}
    0 &= \int_0^T \left( -\ddot{\boldsymbol{g}} +  \boldsymbol{H}_3 \dot{\boldsymbol{g}} + \boldsymbol{H}_2\boldsymbol{H}_1 \boldsymbol{g} + \boldsymbol{h} \right) \cdot \boldsymbol{v} \,dt \\
    &= \dot{\boldsymbol{g}}\cdot \boldsymbol{v} \Big|_{t=0}^{t=T} + b\left( \boldsymbol{g},\boldsymbol{v} \right) + \int_0^T \boldsymbol{h}\cdot \boldsymbol{v}\,dt \\
    &= \dot{g}_3(T)v_3(T) - \dot{g}_3(0)v_3(0).
\end{align*}
The arbitrariness of $v_3(0)$ and $v_3(T)$ implies $\dot{g}_3(0) = \dot{g}_3(T) = 0.$

Finally, we notice that there can be only one $C^2$ solution $\boldsymbol{g}$ of \eqref{eq:ODE_HierarchicalMkt_translated} with BC \eqref{eq:BC_HierarchicalMkt_homog}. Indeed, any function solving this problem must solve the variational one, \eqref{eq:VariationalProblem}, which we proved that admits a single solution.  
\end{proof}

We present an illustration of the current situation in Figure \ref{fig:Hierarchy}. We obtained it by numerically solving the ODE system \eqref{eq:ODE_HierarchicalMkt} with BC \eqref{eq:BC_HierarchicalMkt}. The parameters we employed in this simulation are in Table \ref{tab:Hierarchy}. 

\begin{table}[!htp]
\centering
\begin{tabular}{@{}ccccccc@{}}
\toprule
$q^0_0$ & $\alpha^0$        & $\kappa^0$           & $\lambda^0$        & $\alpha$           & $\kappa$             & $\lambda$           \\ \midrule
 $1$    & $5\times 10^{-5}$ & $2.5 \times 10^{-5}$ & $5 \times 10^{-4}$ & $3 \times 10^{-5}$ & $1.5 \times 10^{-5}$ & $5 \times 10^{-7}$  \\ \bottomrule
\end{tabular}
\caption{Parameters we used in the simulation of the hierarchical market. We fixed $T=1$ and $N=5.$}
\label{tab:Hierarchy}
\end{table}

\begin{figure}[!htp]
    \centering
    \includegraphics[scale = 0.4]{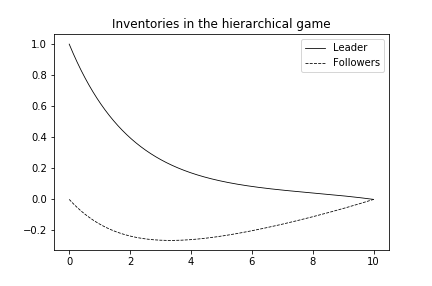}
    \caption{Leader's inventory evolution, and the dynamics of the average holdings of the followers.}
    \label{fig:Hierarchy}
\end{figure}

\section{Conclusions} \label{sec:conclusions}

We analyzed a finite population stochastic differential game of optimal trading. We allowed asymmetry of information, as well as for parameters to be stochastic in our model, with mild assumptions. Preferences were completely heterogeneous among the agents, and differences in their choices were not only due to informational asymmetry. In the first two parts of the work, we investigated Nash equilibria in two settings: the unconstrained and the constrained ones. In the third one, we extended the basic model to a hierarchical market. We introduced hierarchy by assuming there was a leader, also called a leader, having a first-mover advantage, apart from the population of minor traders, which we referred to as followers.

In the unconstrained setting, we characterized the Nash equilibrium as a solution of a coupled vector FBSDE. Employing a continuation technique, we proved that the latter system had a unique solution, under weak interaction. Therefore, given market parameters and initial inventories and cash amounts, a unique corresponding Nash equilibrium exists. Assuming parameters to be constant, we deduced that the FBSDE becomes an ODE, and we obtained semi-explicit forms for it. With the further assumption of homogeneous parameters, we could compare the resulting average inventory with its MFG counterpart.

For the constrained problem, we proved --- under weak interaction and continuity assumptions --- that Nash equilibria are characterized by an FBSDE, which is quite close to that of the unconstrained one, differing only on the terminal condition. We obtained bounds on the solution and corresponding inventory process, uniformly on the terminal penalization parameter, under a more stringent weak interaction assumption. Using functional analytic arguments, we established that the limit was the constrained Nash equilibrium. Putting ourselves under a suitable framework, we proved the convergence of the average speed of trading of the Nash equilibrium of the finite population game to the mean-field optimal aggregation rate, as the population size tends to infinity. We also provided a rate of convergence. 

We proceeded in two steps to study the hierarchical game. Firstly, we assumed that the leader's strategy was given and derived properties of the followers' Nash equilibrium in terms of it. It was appropriate to factor this Nash equilibrium in two parts: one of pure arbitrage and another concerning inventory execution. In the sequel, we feedback the Nash equilibrium of the population, in terms of the leading strategy, in the leader's objective functional. In this way, we rendered the major problem into a single-player optimization. We proved the existence and uniqueness of optimal control, under a suitable weak interaction assumption. We finished this part by discussing the case of constant parameters, assuming homogeneity among followers.

\section*{Acknowledgements}

Y. Thamsten was financed in part by Coordena\c{c}\~ao de Aperfei\c{c}oamento de Pessoal de N\'ivel Superior - Brasil (CAPES) - Finance code 001.


\bibliographystyle{apalike}
\bibliography{refs}


\end{document}